\documentclass[10pt, journal, a4paper, twocolumn]{IEEEtran}

\usepackage{amsmath,amsthm}
\usepackage{amsfonts,amssymb}
\usepackage{bm}
\usepackage{epsfig}
\usepackage{graphicx}
\usepackage{epstopdf}
\usepackage{psfrag}
\usepackage{cite}
\usepackage{url}
\usepackage{subfigure}
\usepackage{algorithm}
\usepackage{algorithmic}

\newtheorem{lemma}{Lemma}

\newcommand{\C}{\mathbf{C}}
\newcommand{\CAB}{\mathbf{C}_{A\rightarrow B}}
\newcommand{\CBA}{\mathbf{C}_{B\rightarrow A}}

\newcommand{\HAB}{\mathbf{H}_{A\rightarrow B}}
\newcommand{\HBA}{\mathbf{H}_{B\rightarrow A}}

\newcommand{\h}{\mathbf{h}}

\newcommand{\T}{\mathbf{T}}
\newcommand{\R}{\mathbf{R}}
\newcommand{\TA}{\mathbf{T}_A}
\newcommand{\TB}{\mathbf{T}_B}
\newcommand{\RA}{\mathbf{R}_A}
\newcommand{\RB}{\mathbf{R}_B}

\newcommand{\Y}{\mathbf{Y}}

\newcommand{\y}{\mathbf{y}}

\newcommand{\f}{\mathbf{f}}
\newcommand{\g}{\mathbf{g}}
\newcommand{\F}{\mathbf{F}}
\newcommand{\FA}{\mathbf{F}_A}
\newcommand{\FB}{\mathbf{F}_B}

\newcommand{\n}{\mathbf{n}}
\newcommand{\N}{\mathbf{N}}
\newcommand{\Lbold}{\mathbf{L}}
\newcommand{\M}{\mathbf{M}}

\newcommand{\I}{\mathbf{I}}
\newcommand{\0}{\mathbf{0}}
\newcommand{\A}{\mathbf{A}}
\newcommand{\rw}{\rightarrow}

\newcommand{\beq}{\begin{equation}}
\newcommand{\eeq}{\end{equation}}
\newcommand{\bal}{\begin{aligned}}
\newcommand{\eal}{\end{aligned}}
\newcommand{\bbm}{\begin{bmatrix}}
\newcommand{\ebm}{\end{bmatrix}}

\newcommand{\diag}{\mbox{diag}}
\newcommand{\B}{\mathbf{B}}
\newcommand{\D}{\mathbf{D}}
\newcommand{\X}{\mathbf{X}}
\newcommand{\x}{\mathbf{x}}
\newcommand{\bmcY}{\bm{\mathcal{Y}}}
\newcommand{\bmcN}{\bm{\mathcal{N}}}
\newcommand{\eo}{\mathbf{e}}
\newcommand{\tin}{\widetilde{\n}}
\newcommand{\tiN}{\widetilde{\N}}
\newcommand{\J}{\mathbf{J}}
\newcommand{\bmcP}{\bm{\mathcal{P}}}
\newcommand{\bmcF}{\bm{\mathcal{F}}}
\newcommand{\bmcH}{\bm{\mathcal{H}}}
\newcommand{\bP}{\mathbf{P}}
\newcommand{\fh}{\hat{\f}}
\newcommand{\hh}{\hat{\h}}

\newcommand{\cC}{{\cal C}}
\newcommand{\cV}{{\mathcal V}}

\IEEEoverridecommandlockouts

\title{A Framework for Over-the-air Reciprocity Calibration for TDD Massive MIMO Systems}

\author{Xiwen ˜JIANG, ˜\IEEEmembership{} 
Alexis ˜Decurninge, ˜\IEEEmembership{}
Kalyana ˜Gopala, ˜\IEEEmembership{} 
Florian ˜Kaltenberger, ˜\IEEEmembership{Member, ˜IEEE,} \\
Maxime ˜Guillaud, ˜\IEEEmembership{Senior Member, ˜IEEE,}
Dirk ˜Slock, ˜\IEEEmembership{Fellow, ˜IEEE,} 
and ˜Luc ˜Deneire, ˜\IEEEmembership{Member, ˜IEEE}%
       \thanks{This work is supported in part by Huawei Mathematical and Algorithmic Sciences Lab in Paris through the project of ``Modeling, Calibrating and Exploiting Channel Reciprocity for Massive MIMO". It is partly supported by the French Government (National Research Agency, ANR) through the ``Investments for the Future" Program \#ANR-11-LABX-0031-01). }
       \thanks{X. Jiang, K. Gopala, F. Kaltenberger and D. Slock are with Communication Systems Department, EURECOM. (e-mail: \{xiwen.jiang, kalyana.gopala, florian.kaltenberger, dirk.slock\}@eurecom.fr)}
       \thanks{A. Decurninge and M. Guillaud are with the Mathematical and Algorithmic Sciences Lab, Paris Research Center, Huawei Technologies France. (e-mail: \{alexis.decurninge, maxime.guillaud\}@huawei.com)}
       \thanks{L. Deneire is with Laboratoire I3S, Universit\'e de Nice Sophia Antipolis. (e-mail: luc.deneire@unice.fr)}
}

\begin{document}

\onecolumn
\vspace*{6cm}
{\Large This work has been submitted to the IEEE for possible publication.  Copyright may be transferred without notice, after which this version may no longer be accessible.}
\newpage
\twocolumn
\maketitle

\begin{abstract}
One of the biggest challenges in operating massive multiple-input multiple-output systems is the acquisition of accurate channel state information at the transmitter. To take up this challenge, time division duplex is more favorable thanks to its channel reciprocity between downlink and uplink. However, while the propagation channel over the air is reciprocal, the radio-frequency front-ends in the transceivers are not. Therefore, calibration is required to compensate the RF hardware asymmetry.

Although various over-the-air calibration methods exist to address the above problem, this paper offers a unified representation of these algorithms, providing a higher level view on the calibration problem, and introduces innovations on calibration methods. We present a novel family of calibration methods, based on antenna grouping, which improve accuracy and speed up the calibration process compared to existing methods. We then provide the Cram\'er-Rao bound as the performance evaluation benchmark and compare maximum likelihood and least squares estimators. We also differentiate between coherent and non-coherent accumulation of calibration measurements, and point out that enabling non-coherent accumulation allows the training to be spread in time, minimizing impact to the data service. Overall, these results have special value in allowing to design reciprocity calibration techniques that are both accurate and resource-effective.

\end{abstract}

\begin{IEEEkeywords}
Massive MIMO, TDD, channel reciprocity calibration.
\end{IEEEkeywords}

\section{Introduction}
\label{sec:intro}

Massive multiple-input multiple-output (MIMO) is a promising air interface technology for the next generation of wireless communications. With large number of antennas installed at the base station (BS) simultaneously serving multiple user equipments (UEs), massive MIMO can dramatically improve the spectral efficiency of cellular networks \cite{marzetta2010noncooperative}, \cite{larsson2013massive}.

For downlink (DL), one of the fundamental challenges to fully realize the potential of massive MIMO is the acquisition of accurate channel state information at the transmitter (CSIT). Time division duplex (TDD) thus attracts great attention from the research community as it enjoys channel reciprocity between DL and uplink (UL), thanks to which the BS can obtain the CSIT from the channel estimation in the UL. In fact, traditional ways to get CSIT from UE feedback becomes infeasible when the antenna array size at the BS scales up, because of the heavy signaling overhead it incurs in the UL.

Channel reciprocity in TDD systems refers to the fact that the physical over-the-air (OTA) channels are the same for UL and DL \cite{Lorentz1896reciprocity,smith2012electromagnetic} within channel coherence time. However, the channel as seen by the digital baseband processor contains not only the physical OTA channel but also radio frequency (RF) front-ends, including the hardware from digital-to-analog converter (DAC) to transmit antennas at the transmitter (Tx) and the corresponding part, from receiving antennas to analog-to-digital converter (ADC), at the receiver (Rx).  The various impairments to reciprocity can be due to manufacturing variability in the power amplifiers and low-noise amplifiers, different cable lengths across the antennas, imperfect clock synchronization, duplexer response, etc. Due to these, the hardware in the Tx and Rx RF chains are, in general, not identical, and therefore the channel from a digital signal processing point of view is not reciprocal. If not taken into account, these hardware-related asymmetries will cause inaccuracy in the CSIT estimation and, as a consequence, seriously degrade the DL beamforming performance \cite{guey2004modeling, luo2016multi, zhang2015large, jiang2016accurately}.

In order to compensate the hardware asymmetry and restore channel reciprocity, calibration techniques are needed. This topic has been explored long before the advent of massive MIMO. In \cite{bourdoux2003non, nishimori2001automatic, nishimori2014effectiveness, petermann2013multi, benzin2017internal}, it is suggested to add additional hardware components in transceivers which are dedicated to calibration. This method (which we refer to as \emph{absolute} calibration) consists in compensating the Tx and Rx RF asymmetry independently in each transceiver; however this does not appear to be a cost-effective solution. \cite{guillaud2005practical, kaltenberger2010relative, shi2011efficient, kouassi2011performance} thus put forward ``relative'' calibration schemes\footnote{The term \emph{relative} indicates here that the calibration coefficients relate the UL and DL digital channels, as opposed to absolute calibration which relates digital domain and propagation domain versions of a channel.}, where the calibration coefficients are estimated using signal processing methods based on OTA bi-directional channel estimation between BS and UE. Since hardware properties can be expected to evolve slowly, and these coefficients can be obtained in the initialization phase of the system (calibration phase), they can be used later together with the instantaneous UL channel estimate to obtain downlink CSIT.

With the advent of massive MIMO, traditional relative calibration methods are challenged, because they require the UE to feed back a large amount of DL CSI for all BS antennas. It was observed in \cite{Kouassi:13} that the calibration factor at the BS side is the same for all channels from the BS to any UE. This was exploited in \cite{Kouassi:13} to determine the BS side calibration factor of a secondary BS with the cooperation of a secondary UE, allowing beamforming with
 zero-forcing to a primary UE without its collaboration. This idea was then pushed further in a number of OTA self-calibration approaches which only require the exchange of OTA training signals between elements of the BS array.
Indeed, for optimizing multi-user massive MIMO systems, the asymmetry in the number of antennas between the BS and the UEs means that most of the massive MIMO multi-user multiplexing gain can be achieved through BS-side only calibration \cite{R1-091794,R1-091752}. These OTA self-calibration approaches have the advantage that, unlike classical single-link relative calibration, no CSI feedback is involved, since all the elements of the BS array are already connected to the same baseband signal processor. Such ``single-side" or ``internal" calibration methods were proposed in \cite{shepard2012argos, rogalin2014scalable, vieira2014reciprocity, vieira2017reciprocity, papadopoulos2014avalanche}. In \cite{shepard2012argos}, the authors reported on the massive MIMO Argos prototype, where calibration is performed OTA with the help of a reference antenna. By performing bi-directional transmission between the reference antenna and the rest of the antenna array, it is possible to estimate the calibration coefficients up to a common scalar ambiguity which will not influence the final DL beamforming capability. The Argos calibration approach however is sensitive to the location of the reference antenna, and as one of the consequences, is not suitable for distributed massive MIMO. This concern motivated the introduction of a method (Rogalin \emph{et al.} in \cite{rogalin2014scalable}) whereby calibration is not performed w.r.t. a reference antenna\footnote{The method in \cite{rogalin2014scalable} is denoted as ``least-squares (LS) calibration", however we will not use this terminology since most calibration techniques proposed in the literature ultimately rely on LS estimation}. It has the spirit of distributed algorithms, making it a good calibration method for antenna arrays having a distributed topology. Note that it can also be applied to colocated massive MIMO, as  in the LuMaMi massive MIMO prototype \cite{vieira2014flexible} where a weighted version of the estimator presented in \cite{vieira2014reciprocity} is used, whereas a Maximum Likelihood (ML) estimator is presented in \cite{vieira2017reciprocity}.  Moreover, a fast calibration method named Avalanche was proposed in \cite{papadopoulos2014avalanche}; its principle is to use a calibrated sub-array to calibrate uncalibrated elements. The calibrated array thus grows during the calibration process in a way similar to the avalanche phenomenon.

Among other relevant works, we refer to \cite{luo2015robust, wei2016mutual, tsoulos1997calibration, tsoulos1998space, Jiang2015}. In \cite{luo2015robust}, the author provides an idea to perform system health monitoring on the calibrated reciprocity. Under the assumption that the majority of calibration coefficients stay calibrated and only a minority of them change, the authors propose a compressed sensing enabled detection algorithm to find out which calibration coefficient has changed based on the sparsity in the vector representing the coefficient change. In \cite{wei2016mutual}, a calibration method dedicated to maximum ratio transmission (MRT) is proposed. Experimental data about the calibration coefficients are reported in \cite{tsoulos1997calibration, tsoulos1998space, shepard2012argos, Jiang2015}, giving an insight on how the impairments evolve in the time and frequency domains as well as with the temperature, and about the hardware properties behind this effect.

In the present article, we introduce a unified framework to represent different existing calibration methods. Although they appear at first sight to be different, we reveal that all existing calibration methods can be modeled under a general pilot based calibration framework; different ways to partition the array into transmit and receive elements during successive training phases yield different schemes. The unified representation shows the relationship between these methods and provides alternative ways to obtain corresponding estimators.
As this framework gives a general and high level understanding of the TDD calibration problem in massive MIMO systems, it opens up possibilities for new calibration methods. As an example, we present a novel family of calibration schemes based on antenna grouping, which can greatly speed up the calibration process with respect to the classical approaches. We will show that our proposed method greatly outperforms the Avalanche method \cite{papadopoulos2014avalanche} in terms of calibration accuracy, yet it is equally fast. In order to evaluate the performance of calibration schemes, we derive the Cram\'er-Rao bounds (CRB) of the accuracy of calibration coefficients estimation. Another important contribution of this work is the introduction of non-coherent accumulation of the measurements used for calibration. We will see that calibration does not necessarily have to be performed in an intensive manner during a single channel coherence interval, but can rather be executed using time resources distributed over a relatively long period. This enables TDD reciprocity calibration to be interleaved with the normal data transmission or reception, leaving it almost invisible for the whole system. 


The rest of this paper is organized as follows. Section~\ref{sec:tdd_rec_calib} describes the basic principles of reciprocity calibration in a TDD based MIMO system. Section~\ref{sec:uni_rep} presents the TDD reciprocity system model and introduces our unified framework. Section~\ref{sec:prior} presents how Argos, Rogalin and Avalanche calibration algorithms fit into this model as well as how we can obtain the corresponding estimators. In Section~\ref{sec:fast_calib}, we present the fast calibration scheme based on antenna grouping and discuss the minimum number of channel uses it requires to estimate all calibration coefficients. In Section~\ref{sec:crb}, we address the optimal estimation problem of reciprocity calibration parameters, we derive the CRB, propose a maximum likelihood (ML) estimator and compare it with the LS estimator. Section~\ref{sec:non_coh_accum} is dedicated to non-coherent accumulation of measurements. In Section~\ref{sec:simu_result}, we illustrate the performance of the group-based fast calibration method and compare its performance with other calibration algorithms using CRB as the benchmark. Conclusions are drawn in Section~\ref{sec:conclusions}.

The notation adopted in this paper conforms to the following conventions. Vectors and matrices are denoted in lowercase bold and uppercase bold respectively: $\mathbf{a}$ and $\mathbf{A}$. $(\cdot)^*$, $(\cdot)^T$, $(\cdot)^H$, $(\cdot)^{\dagger}$ denote element-wise complex conjugate, transpose, Hermitian transpose and Moore-Penrose pseudo inverse, respectively. $\otimes$ and $*$ denotes the Kronecker product operator and the Khatri--Rao product \cite{khatri1968solutions}, respectively. $\lceil \cdot \rceil$ is the ceiling operator, which rounds a number to the next integer. $\mbox{diag}\{a_1, a_2, \dots, a_M\}$ denotes a diagonal matrix with its diagonal composed of $a_1, a_2, \dots, a_M$, whereas $\mbox{vec}(\mathbf{A})$ denotes the vectorization of the matrix $\mathbf{A}$. $\mathbb{C}$ denotes the set of complex numbers. 

\section{OTA Reciprocity Calibration}
\label{sec:tdd_rec_calib}

In this section, we describe the basic idea of reciprocity calibration in a practical TDD system. Let us consider a system as in Fig.~\ref{fig:rec_mod}, where A represents a BS and B represents a UE, each containing $M_A$ and $M_B$ antennas, respectively. The DL and UL channels flat-fading model (as typically obtained by considering a single subcarrier of a multicarrier system) seen in the digital domain are noted by $\HAB$ and $\HBA$. Since they are formed by the cascade of the Tx impairments, OTA propagation, and Rx impairments, they can be represented by
\beq
\label{eqn:rec_mod1}
\left\{
\bal
 \HAB & =  \RB\CAB\TA,\\
 \HBA & =  \RA\CBA\TB,
\eal
\right.
\eeq
where matrices $\TA$, $\RA$, $\TB$, $\RB$ model the response of the transmit and receive RF front-ends, while $\CAB$ and $\CBA$ model the OTA propagation channels, respectively from A to B and from B to A. The dimension of $\TA$ and $\RA$ are $M_A \times M_A$, whereas that of $\TB$ and $\RB$ are $M_B\times M_B$. The diagonal elements in these matrices represent the linear effects attributable to the impairments in the transmitter and receiver parts of the RF front-ends respectively, whereas the off-diagonal elements correspond to RF crosstalk and antenna mutual coupling\footnote{Here, ``antenna mutual coupling'' is used to describe parasitic effects that two nearby antennas have on each other, when they are either both transmitting or receiving \cite{balanis2016antenna, petermann2013multi}. However, this is different to the channel between transmitting and receiving elements of the same array, which we call the intra-array channel. Note that this differs from the terminology in \cite{vieira2017reciprocity} and \cite{wei2016mutual} where the term mutual coupling is used to denote the intra-array channel.}. It is worth noting that although transmitting and receiving antenna mutual coupling is not generally reciprocal \cite{wei2015reciprocity}, theoretical modeling \cite{petermann2013multi} and experimental results \cite{Jiang2015, shepard2012argos, vieira2017reciprocity} both show that in practice, RF crosstalk and antenna mutual coupling can be ignored for the purpose of reciprocity calibration, which implies that $\TA$, $\RA$, $\TB$, $\RB$ can safely be assumed to be diagonal.

Assuming the system is operating in TDD mode, the OTA channel responses enjoy reciprocity within the channel coherence time, i.e., $\CAB = \CBA^T$. Therefore, we obtain the following relationship between the channels measured in both directions:
\beq
\label{eqn:rec_mod2}
  \bal
    \HAB & = \RB(\RA^{-1}\HBA\TB^{-1})^T\TA \\
            & = \underbrace{\RB\TB^{-T}}_{\FB^{-T}}\HBA^T \underbrace{\RA^{-T}\TA}_{\FA} \\
            & = \FB^{-T}\HBA^T\FA.
  \eal
\eeq
A system utilizing OTA reciprocity calibration normally has two phases for its function. Firstly, during the initialization of the system, the calibration process is performed, which consists in estimating $\FA$ and $\FB$.
Then during the data transmission phase, they are used together with instantaneous measured UL channel $\hat{\mathbf{H}}_{B\rw A}$ to estimate $\HAB$ according to \eqref{eqn:rec_mod2}, based on which advanced beamforming algorithms can be performed.
Since the calibration coefficients typically remain stable \cite{shepard2012argos}, the calibration process does not have to be performed very frequently.

Note that the studies in \cite{R1-091794,R1-091752} pointed out that in a practical multi-user MIMO system, it is mainly the calibration at the BS side which restores the hardware asymmetry and helps to achieve the multi-user MIMO performance, whereas the benefit brought by the calibration on the UE side is not necessarily justified. We thus, in the sequel, focus on the estimation of $\F_A$, although the framework discussed in the following section is not limited to this case.

\begin{figure}[t]
\centering 
\includegraphics[width=\columnwidth]{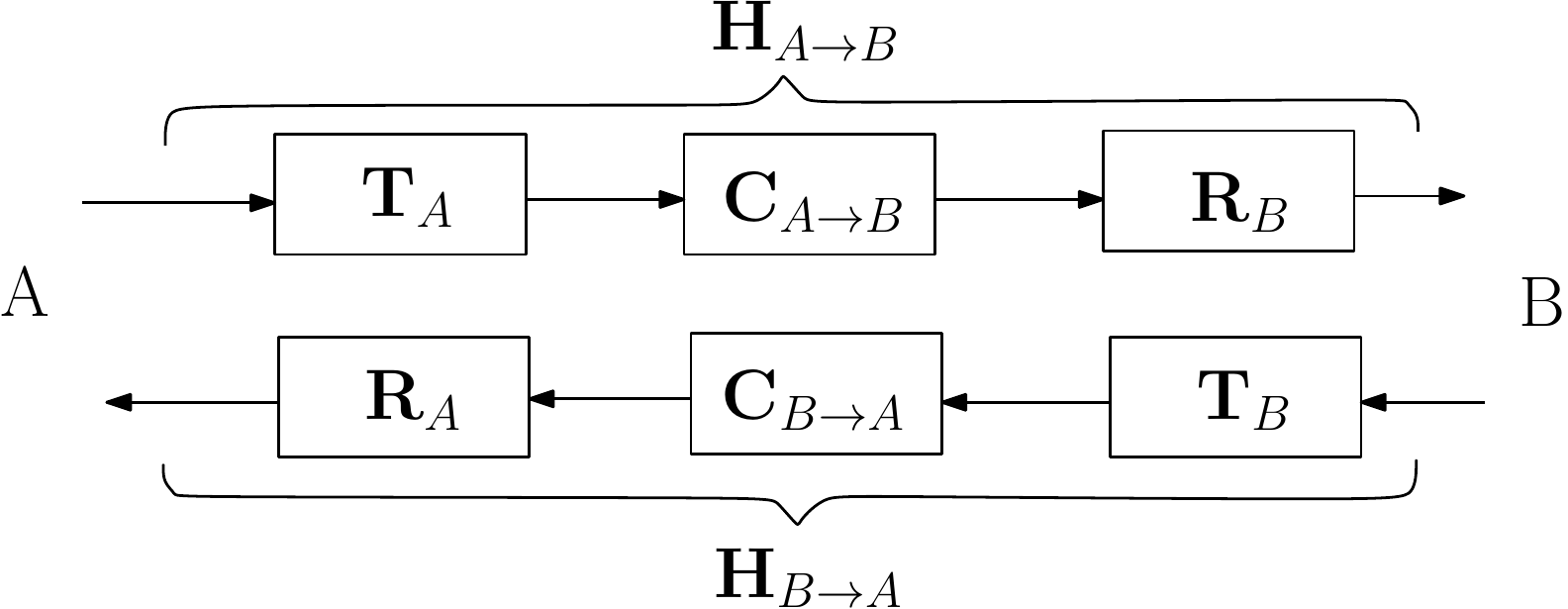}
\caption{Reciprocity Model}
\label{fig:rec_mod}
\end{figure}

\section{General OTA Calibration Framework}
\label{sec:uni_rep}

\subsection{Overview and signalling}
In this section, we present a general framework for OTA pilot-based reciprocity calibration.
Let us consider an antenna array of $M$ elements partitioned into $G$ groups denoted by $A_1, A_2,\dots, A_G$, as in Fig.~\ref{fig:group_calib}. 
Group $A_i$ contains $M_i$ antennas such that $\sum_{i=1}^G M_i = M.$
Each group $A_i$ transmits a sequence of $L_i$ pilot symbols, defined by matrix $\mathbf{P}_i\in\mathbb{C}^{M_i\times L_i}$ where the rows correspond to antennas and the columns to successive channel uses. Note that a channel use can be understood as a time slot or a subcarrier in an OFDM-based system, as long as the calibration parameter can be assumed constant over all channel uses. When an antenna group $i$ transmits, all other groups are considered in receiving mode.  After all $G$ groups have transmitted, the received signal for each resource block of bidirectional transmission between antenna groups $i$ and $j$ is given by
\begin{equation}
\label{eqn:sig_mod3}
\left\{
\begin{array}{l}
\Y_{i\rw j}  = \R_j\C_{i\rw j}\T_i\mathbf{P}_i+ \N_{i\rw j},
\\
\Y_{j\rw i}  = \R_i\C_{j\rw i}\T_j\mathbf{P}_j + \N_{j\rw i},
\end{array}
\right.
\end{equation}
where $\Y_{i\rw j} \in\mathbb{C}^{M_j\times L_i}$ and $\Y_{j\rw i}\in\mathbb{C}^{M_i\times L_j}$ are received signal matrices at antenna groups $j$ and $i$ respectively when the other group is transmitting. $\N_{i\rw j}$ and $\N_{j\rw i}$ represent the corresponding received noise matrix. $\T_i,\;\R_i \in \mathbb{C}^{M_i\times M_i}$ and $\T_j,\;\R_j \in \mathbb{C}^{M_j\times M_j}$ represent the effect of the transmit and receive RF front-ends of antenna elements in groups $i$ and $j$ respectively.

The reciprocity property induces that $\C_{i\rw j}=\C_{j\rw i}^T$,  thus for two different groups $1\leq i\neq j \leq G$, by eliminating $\C_{i\rw j}$ in \eqref{eqn:sig_mod3} we have
\begin{equation}
\label{eqn:obs_general}
\mathbf{P}_i^T\F_i^{T}\Y_{j\rw i} - \Y_{i\rw j}^T\F_j\mathbf{P}_j = \tiN_{ij},
\end{equation}
where the noise component $\tiN_{ij}=\mathbf{P}_i^T\F_i^{T}\N_{j\rw i} - \N_{i\rw j}^T\F_j\mathbf{P}_j$, while $\F_i=\R_i^{-T}\T_i$ and $\F_j = \R_j^{-T}\T_j$ are the calibration matrices for groups $i$ and $j$. The calibration matrix $\F$  is diagonal, and thus takes the form of
$\F = \mbox{diag}\{\F_1, \F_2, \dots, \F_G\}$.
Note that estimating $\F_i$ or $\F_j$ from \eqref{eqn:obs_general} for a given pair $(i,j)$ does not exploit all relevant received data. An optimal estimation jointly considering all received signals for all $(i,j)$ will be proposed in Section~\ref{sec:crb}. Note that the proposed framework also allows to consider using only subsets of the received data which corresponds to some of the methods found in the literature.

\begin{figure}[!t]
\centering
\includegraphics[width=\columnwidth]{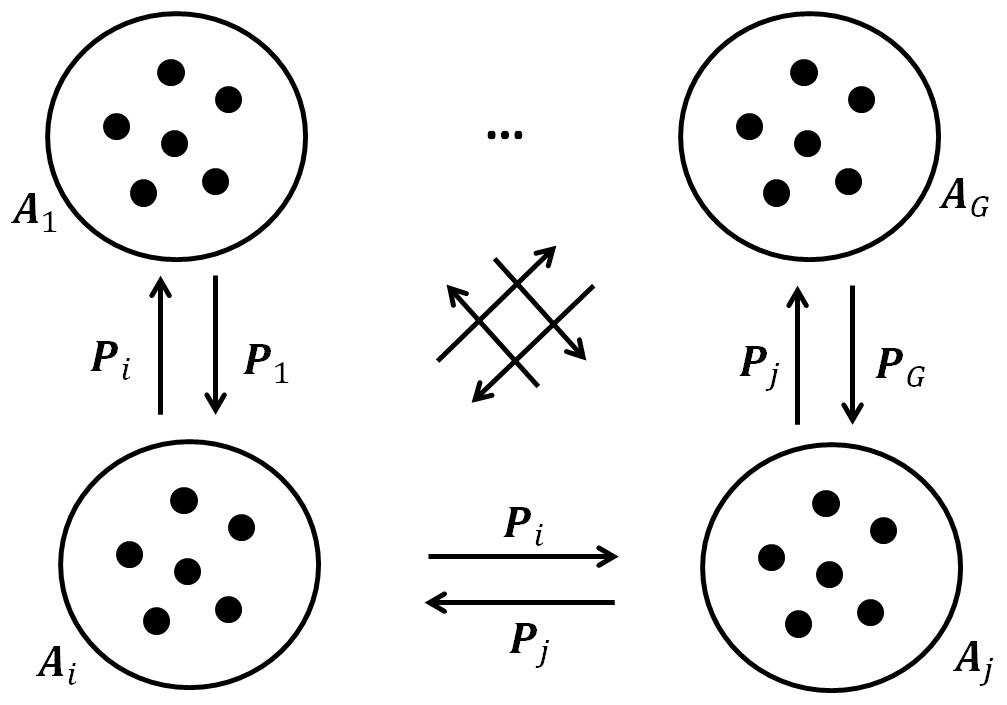}
\caption{Bi-directional transmission between antenna groups.}
\label{fig:group_calib}
\end{figure}

Let us use $\f_i$ and $\f$ to denote the vectors of the diagonal coefficients of $\F_i$ and $\F$ respectively, i.e., $\F_i = \diag\{\f_i\}$ and $\F = \diag\{\f\}$.
This allows us to vectorize (\ref{eqn:obs_general}) into
\begin{equation}
\label{eqn:obs_general2}
(\Y_{j\rw i}^T*\mathbf{P}_i^T)\f_i - (\mathbf{P}_j^T*\Y_{i\rw j}^T)\f_j = \tin_{ij},
\end{equation}
where $*$ denotes the Khatri--Rao product (or column-wise Kronecker product\footnote{With matrices $\A$ and $\B$ partitioned into columns, $\mathbf{A} = \bbm\mathbf{a}_1 & \mathbf{a}_2 & \dots &\mathbf{a}_M\ebm$ and $\mathbf{B} = \bbm \mathbf{b}_1 & \mathbf{b}_2 & \dots &\mathbf{b}_M\ebm$ where $\mathbf{a}_i$ and $\mathbf{b}_i$ are column vectors for $i \in 1 \dots M$, then, $\mathbf{A}*\mathbf{B} = \bbm \mathbf{a}_1\otimes \mathbf{b}_1 & \mathbf{a}_2\otimes \mathbf{b}_2 & \dots & \mathbf{a}_M\otimes \mathbf{b}_M\ebm$ \cite{khatri1968solutions}.}), where we have used the equality $\mbox{vec}(\A\,\diag(\x)\,\B) = (\B^T * \A) \, \x$. Note that, if we do not suppose that every $\F_i$ is diagonal, (\ref{eqn:obs_general2}) holds more generally by replacing the Katri--Rao products by Kronecker products and $\f_i$ by $\mbox{vec}(\F_i)$.
Finally, stacking equations \eqref{eqn:obs_general2} for all $1\leq i<j\leq G$ yields
\begin{equation}
\label{eqn:lineq_calib_reduced}
\bmcY(\mathbf{P})\f =\tin,
\end{equation}
with $\bmcY(\mathbf{P})$ defined as
\begin{equation}
\label{eqn:cursivey}
\underbrace{\begin{bmatrix}
(\Y_{2 \rw 1}^{T} * \mathbf{P}_1^{T}) &  -(\mathbf{P}_2^{T} * \Y^T_{1 \rw 2}) & 0 & \dots \\
(\Y_{3 \rw 1}^{T} * \mathbf{P}_1^{T}) & 0 &  -(\mathbf{P}_3^{T} * \mathbf{Y}^{T}_{1 \rw 3}) & \dots \\
0 &  (\Y_{3 \rw 2}^{T} * \mathbf{P}_2^{T}) &  -(\mathbf{P}_3^{T} * \mathbf{Y}^{T}_{2 \rw 3})   & \dots \\ 
\vdots & \vdots & \vdots & \ddots
\end{bmatrix}}_{(\sum_{j=2}^{G} \sum_{i=1}^{j-1} L_i L_j ) \times  M}.
\end{equation}

It is worth noting that this framework is not limited to represent single-side calibration. For UE-aided (relative) calibration, it suffices to set $2$ groups such as $A_1$ and $A_2$, representing the BS and the UE, respectively in order to get a full calibration scheme.

\subsection{Parameter identifiability and pilot design}
\label{sec:identifiability}

Before proposing an estimator for $\f$, we raise the question of the problem identifiability which corresponds to the fact that \eqref{eqn:lineq_calib_reduced} admits a unique solution in the noiseless scenario
\begin{equation}
\bmcY(\mathbf{P})\f = \mathbf{0}.
\label{eqn:identifiability}
\end{equation}
The solution of \eqref{eqn:identifiability} is defined up to a complex scalar factor $\alpha$, since if $\f$ is a solution, then $\alpha \f$ is also a solution of \eqref{eqn:identifiability}.
This indeterminacy can be resolved by fixing one of the calibration parameters, say $f_1 = \eo_1^H \f = [1\, 0\, \cdots\, 0] \f = 1$ or by a norm constraint, for example $\|\f\| = 1$.
Then, the identifiability is related to the dimension of the kernel of $\bmcY(\mathbf{P})$ in the sense that the problem is fully determined if and only if the kernel of $\bmcY(\mathbf{P})$ is of dimension $1$. Since the true $\f$ is a solution to \eqref{eqn:identifiability}, we know that the rank of $\bmcY(\mathbf{P})$ is at most $M-1$. We will assume furthermore in the following that the pilot design is such that the rows of $\bmcY(\mathbf{P})$ are linearly independent as long as the number of rows is less than $M-1$. Note that this condition depends on the internal channel realization $\C_{i \rw j}$ and on the pilot matrices $\mathbf{P}_i$. However, sufficient conditions of identifiability expressed on these matrices are out of the scope of this paper. 
Under rows independence, \eqref{eqn:lineq_calib_reduced} may be read as the following sequence of events:
\begin{enumerate}
\item Group $1$ broadcasts its pilots to all other groups using $L_1$ channel uses;
\item After group 2 transmits its pilots, we can formulate $L_2L_1$  equations of the form \eqref{eqn:obs_general2};
\item After group 3 transmits its pilots, we can formulate $L_3L_1 + L_3L_2$ equations;
\item After group j transmits its pilots, we can formulate $\sum_{i=1}^{j-1}L_jL_i$ equations.
\end{enumerate}
This process continues until group $G$ finishes its transmission, and the whole calibration process finishes.
During this process of transmission by the $G$ antenna groups, we can start forming equations as indicated, that can be solved recursively for subsets of unknown calibration parameters, or we can wait until all equations are formed to solve the problem jointly.
By independence of the rows, we can state that the problem is fully determined if and only if 
\beq
\sum_{1\leq i < j \leq G} L_j L_i   \ge M-1\; .
\label{eqident}\eeq

\subsection{LS calibration parameter estimation}
\label{sec_LS_estimation}

A typical way to estimate $\f$ consists in solving a LS problem such as
\begin{equation}
\label{eqn:general_LS_min}
\mbox{}\!\!\!\!\!\!\!\!
\begin{array}{lll}
\hat{\f} &\!\!\!\! =\!\!\!\! &\displaystyle\arg\min_{\f}\|\bmcY(\mathbf{P})\,\f\|^2\\
 &\!\!\!\! =\!\!\!\! &\displaystyle\arg\min_{\f} \displaystyle\sum_{i<j} \|(\Y_{j\rw i}^T*\mathbf{P}_i^T)\f_i - (\mathbf{P}_j^T*\Y_{i\rw j}^T)\f_j\|^2
\end{array},
\end{equation}
where $\bmcY(\mathbf{P})$ is defined in \eqref{eqn:cursivey}. 
This needs to be augmented with a constraint
\beq
\cC(\fh,\f)=0,
\label{eqn:genconstr}
\eeq
in order to exclude the trivial solution $\fh=\0$ in \eqref{eqn:general_LS_min}.
The constraint on $\fh$ may depend on the true parameters $\f$.
As we shall see further this constraint needs to be complex valued (which represents two real constraints).
Typical choices for the constraint are\\
1) Norm plus phase constraint (NPC):
\begin{eqnarray}
 & &\mbox{}\!\!\!\!\!\!\!\!\!\!\!\!\!\!\!\!\!\!
\mbox{norm: } \mbox{Re}\{\cC(\fh,\f)\}= ||\hat{\f}||^2 - c\, , \; c= ||\f||^2 ,
\label{eqRe}\\
 & &\mbox{}\!\!\!\!\!\!\!\!\!\!\!\!\!\!\!\!\!\!
\mbox{phase: } 
\mbox{Im}\{\cC(\fh,\f)\}=
\mbox{Im}\{\fh^H\f\}
 = 0.
\label{eqIm}
\end{eqnarray}
2) Linear constraint:
\beq
\cC(\fh,\f)= \hat{\f}^H\g -c = 0\; .
\label{eqn:LinConstr}
\eeq
If we choose the vector $\g=\f$ and $c= ||\f||^2$, then the $\mbox{Im}\{.\}$ part of \eqref{eqn:LinConstr} corresponds to \eqref{eqIm}.
The most popular linear constraint is the First Coefficient Constraint (FCC), which is \eqref{eqn:LinConstr} with $\g= \eo_1$, $c=1$.
The solution of \eqref{eqn:general_LS_min}, \eqref{eqn:LinConstr} is given by
\begin{equation}
\begin{array}{lll}
\fh&=& \arg \displaystyle\min_{\f : \f^H\g = c}\, \|\bmcY(\mathbf{P})\,\f\|^2 \\
&=& \displaystyle\frac{c}{\g^H (\bmcY(\mathbf{P})^H\bmcY(\mathbf{P}))^{-1}\g} (\bmcY(\mathbf{P})^H\bmcY(\mathbf{P}))^{-1}\g \, .
\label{eq:eqn_basic_vec6}
\end{array}
\end{equation}
Assuming a unit norm constraint (\eqref{eqRe} with $c=1$) on the other hand yields
\beq
\fh' = \arg\min_{\f : \|\f\| = 1} \|\bmcY(\mathbf{P})\,\f\|^2 = V_{min}(\bmcY(\mathbf{P})^H\bmcY(\mathbf{P})),
\label{eq:eqn_basic_vec7}
\eeq
where $V_{min}(\X)$ denotes the eigenvector of matrix $\X$ corresponding to its eigenvalue with the smallest magnitude.
Then the NPC solution of \eqref{eqn:general_LS_min}, \eqref{eqRe}, \eqref{eqIm} is $\fh = \sqrt{c}\, e^{j\phi} \fh'$ in which the phase $\phi$ is adjusted to satisfy \eqref{eqIm}, i.e. $\phi = \arg(\fh^{'H}\f)$ where for any complex number $z = |z| e^{j\arg(z)}$.

\section{Existing calibration techniques}
\label{sec:prior}

Different choices for the partitioning of the $M$ antennas and the pilots matrices exposed in Section~\ref{sec:uni_rep} lead to different calibration algorithms. 
We will now see how different estimators of the calibration matrix can be derived from \eqref{eqn:obs_general2}. In order to ease the description, we assume that the channel is constant during the whole calibration process, this assumption will later be relaxed and discussed in Section~\ref{sec:non_coh_accum}.

\subsection{Argos}
\label{subsec:Argos}

The Argos calibration method \cite{shepard2012argos} consists in performing bi-directional transmission between a carefully chosen reference antenna and the rest of the antenna array. This can be recast in our framework by considering $G=2$ sets of antennas, with set $A_1$ containing only the reference antenna ($M_1 = 1$), and set $A_2$ containing all the other antenna elements ($M_2 = M-1$), as shown in Fig.~\ref{fig:argos_calib}. Firstly, pilot $1$ is broadcasted from the reference antenna to all antennas in set $A_2$, thus $L_1=1$, $\mathbf{P}_1 = 1$ and $\f_2 = \bbm f_2,\dots, f_M\ebm^T$. Then, antennas in set $A_2$ successively transmit pilot $1$ to the reference antenna, thus $L_2=M-1$ and $\mathbf{P}_2 =\mathbf{I}_{M-1}$. \eqref{eqn:obs_general2} thus becomes
\beq
\label{eqn:obs_argos}
f_1\y_1 = \mbox{diag}(\y_2)\mathbf{f}_2 + \tilde{\n},
\eeq
where $\y_1 = \bbm y_{2\rw 1}& y_{3\rw 1}& \dots & y_{M\rw 1}\ebm^T$ and $\y_2 = \bbm y_{1\rw 2} & y_{1\rw 3} &  \dots &  y_{1\rw M}\ebm^T$ with $y_{i\rw j}$ representing the signal transmitted from antenna $i$ and received at antenna $j$. \eqref{eqn:obs_argos} can be decomposed into $M-1$ independent equations as $f_1 y_{i\rw 1} = f_i y_{1\rw i} + \tilde{n}_i$,
where $\tilde{n}_i$ is the $i^{th}$ element in the noise vector $\tilde{\n}$. The LS estimator for each element is thus
\beq
\label{eqn:estimator_argos}
f_i = f_1\frac{y_{i\rw 1}}{y_{1\rw i}}, \;\;\mbox{where} \;\; i = 2,3,\dots,M.
\eeq

\begin{figure}[!t]
\centering
\includegraphics[width=\columnwidth]{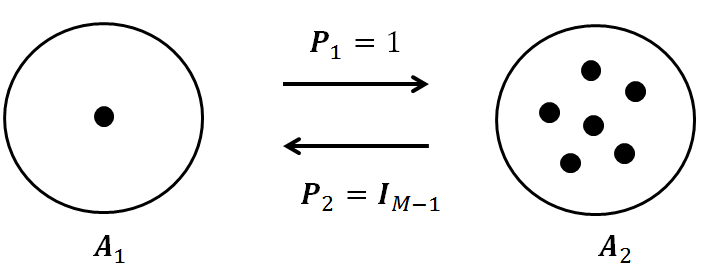}
\caption{Argos calibration}
\label{fig:argos_calib}
\end{figure}

\subsection{Methods based on successive single-antenna transmissions followed by joint estimation}
\label{subsec:LS}

The method from Rogalin et al. presented in \cite{rogalin2013hardware,rogalin2014scalable} and further analyzed in \cite{vieira2017reciprocity} is based on single-antenna transmission at each channel use; all received signals are subsequently taken into account through joint estimation of the calibration parameters. In order to represent this method within the unified framework, we define each set $A_i$ as containing only antenna $i$, i.e., $M_i = 1$ for $1\leq i\leq M$, as in Fig.~\ref{fig:rogalin_calib}.  
Since we assume that the channel is constant, this calibration procedure can be performed in a way that antennas can broadcast pilot 1 in a round-robin manner to all other antennas. In total, $M$ channel uses are needed to finish the transmission, making the pilots to be $\mathbf{P}_i=1$ (with $L_i=1$). With these pilot exchanges, \eqref{eqn:obs_general2} degrades to
\begin{equation}
\label{eqn:obs_LS}
y_{j\rw i}f_i - y_{i\rw j}f_j = \tilde{n}.
\end{equation}
Estimating the calibration coefficients can be performed using \eqref{eq:eqn_basic_vec6} or \eqref{eq:eqn_basic_vec7}. Let us use $\A$ to denote $\bmcY(\mathbf{P})^H\bmcY(\mathbf{P})$, its element on the $i^{th}$ row and $j^{th}$ column is then given by
\begin{equation}
A_{i,j} = \left\{
  \bal
  &\sum_{k\neq i}|y_{k\rightarrow i}|^2 &\mbox{for } & j = i, \\
  &-y_{j\rightarrow i}^*y_{i\rightarrow j} &\mbox{for } & j\neq i.
  \eal
\right.
\end{equation}
Assuming a unit norm constraint, the solution given by $V_{min}(\A)$ matches that in \cite{rogalin2014scalable} whereas the solution under FFC corresponds to that given in \cite{rogalin2013hardware}. Note, however, that calibration coefficients in \cite{rogalin2014scalable, rogalin2013hardware} are defined as the inverse of the $f_i$ in the current paper. 

\begin{figure}[!t]
\centering
\includegraphics[width=0.8\columnwidth]{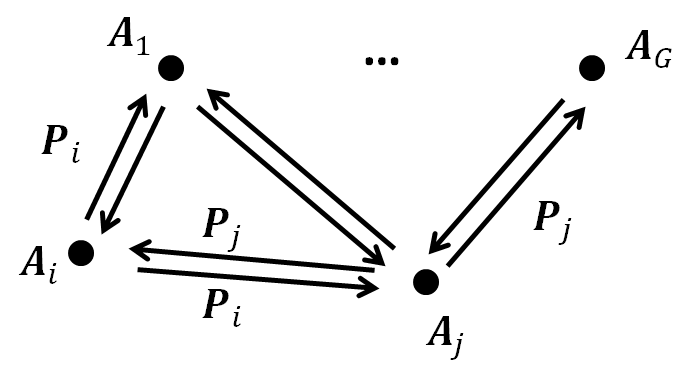}
\caption{Method of Rogalin et al. for reciprocity calibration. Not all links between elements are plotted.}
\label{fig:rogalin_calib}
\end{figure}

Other methods following the same single antenna partition scenario can be viewed as variants of the method above. For example, by allowing only the transmission between two neighboring antennas (antenna index difference is 1), \eqref{eqn:obs_LS} becomes $ f_iy_{i-1\rw i} = f_{i-1}y_{i\rw i-1} + \tilde{n}$. Thus, $ f_i = \frac{y_{i\rw i-1}}{y_{i-1 \rw i}} f_{i-1} + \tilde{n}$. By setting the first antenna as the reference antenna with $f_1=1$, we can obtain a daisy chain calibration method as in \cite{benzin2017internal}, although the original was presented as a hardware-based calibration.
Another variant considered in \cite{vieira2014reciprocity} consists in weighting the error metric such as $|\beta_{j\rw i}f_iy_{j\rw i} - \beta_{i\rw j}f_jy_{i\rw j}|^2$ where the weights $\beta_{j\rw i}$ and $\beta_{i\rw j}$ are based on the SNR of the intra-array channel between antenna element $i$ and $j$.

\subsection{Avalanche}
\label{subsec:Avalanche}

Avalanche \cite{papadopoulos2014avalanche} is a family of fast recursive calibration methods. 
The algorithm successively uses already calibrated parts of the antenna array to calibrate uncalibrated antennas which, once calibrated, are merged into the calibrated array. A full Avalanche calibration may be expressed under the unified framework by considering $M=\frac{1}{2}G(G-1)+1$ antennas where $G$ is the number of groups of antennas partitioning the set of antenna elements as follows: group $A_1$ contains antenna 1, group $A_2$ contains antenna 2, group $A_3$ contains antennas 3 and 4, etc. until group $A_G$ that contains the last $G-1$ antennas. In other terms, group $A_i$ contains $M_i=\max(1,i-1)$ antennas. Moreover, in the method proposed in \cite{papadopoulos2014avalanche}, each group $A_i$  uses $L_i=1$ channel use by sending the pilot $\mathbf{P}_i = \mathbf{1}_{M_i\times 1}$. An example with $7$ antenna elements partitioned into $4$ antenna groups, where we use group $1$, $2$, $3$ (assumed to be already calibrated) to calibrate group $4$, is shown in Fig.~\ref{fig:avalanche_calib}.
In this case, \eqref{eqn:obs_general2} becomes
\beq
\label{eqn:obs_avalanche}
(\y_{j\rw i}^T*\mathbf{P}_i^T)\f_i - (\mathbf{P}_j^T*\y_{i\rw j}^T)\f_j = \tin_{ij}.
\eeq
In \cite{papadopoulos2014avalanche}, the authors exploited an online version of the LS estimator using previously estimated calibration parameters $\mathbf{\hat{f}}_1,\dots,\mathbf{\hat{f}}_{i-1}$ by minimizing
\begin{eqnarray}
\label{eqn:est_avalanche}
\mathbf{\hat{f}}_i &=& \underset{\mathbf{f}_i}{\arg\min} \sum_{j=1}^{i-1} \left\|(\y_{j\rw i}^T*\mathbf{P}_i^T)\f_i - (\mathbf{P}_j^T*\y_{i\rw j}^T)\hat{\f}_j\right\|^2 \nonumber \\
&=& (\Y_i^H\Y_i)^{-1}\Y_i^H\mathbf{a}_i,
\end{eqnarray}
where
$\Y_i = \bbm \mathbf{y}_{1\rw i} & \mathbf{y}_{2\rw i} &\dots & \mathbf{y}_{i-1\rw i} \ebm^T\in\mathbb{C}^{(i-1)\times M_i}$,
and
$\mathbf{a}_i = [(\mathbf{P}_1^T*\y_{i\rw 1}^T)\hat{\f}_1,\dots, (\mathbf{P}_{i-1}^T*\y_{i\rw i-1}^T)\hat{\f}_{i-1}]\in\mathbb{C}^{(i-1)\times 1}$.
Two things should be noted, firstly, $\f_{1},\dots,\f_{i-1}$ are replaced by their estimated version which causes error propagation: estimation errors on a given calibration coefficient will propagate to subsequently calibrated antenna elements. Secondly, in order for \eqref{eqn:est_avalanche} to be well-defined, i.e., in order for $\Y_i^H\Y_i$ to be invertible, it is necessary that $M_i \leqslant i-1$. Note that this necessary condition is specific to the considered online LS estimator \eqref{eqn:est_avalanche} and is more restrictive than the identifiability condition exposed in Section \ref{sec:identifiability}.

\begin{figure}[!t]
\centering
\includegraphics[width=\columnwidth]{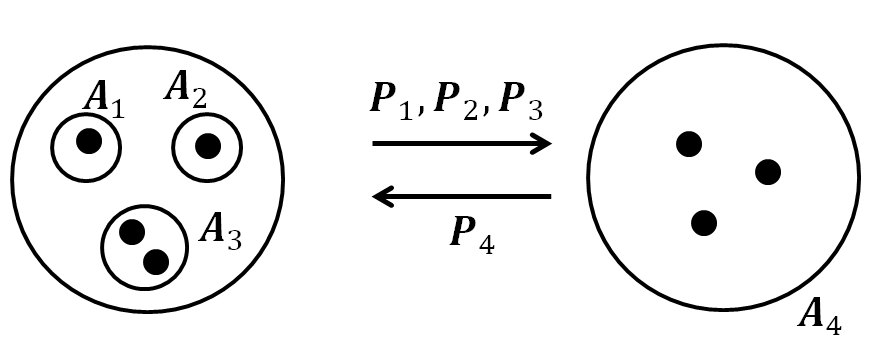}
\caption{Example of full Avalanche calibration with $7$ antennas partitioned into $4$ groups. Group $1$, $2$, $3$ have already been calibrated, and group $4$ is to be calibrated.}
\label{fig:avalanche_calib}
\end{figure}

\section{Fast Calibration: optimal antenna grouping}
\label{sec:fast_calib}

The general calibration framework in Section~\ref{sec:uni_rep} opens up possibilities for new calibration schemes by using new ways to group up antennas.
In this section we show that considering groups of antennas can potentially reduce the total number of channel uses necessary for calibration; we derive the theoretical limit on the smallest number of groups (and associated channel uses) needed to perform calibration.

We first address the problem of finding the smallest number of groups enabling calibration of the whole array while ensuring identifiability at each step,  by finding the best choices for the $L_{i}$ in order to see to what extent optimizing the group based calibration can speed up the calibration process. Let us consider the case where the total number of channel uses available for calibration is fixed to $K$. We derive the number of pilot transmissions for each group, $L_1,\dots,L_G$, that would maximize the total number of antennas that can be calibrated, i.e.,
\begin{equation}
\max_{(L_1,\dots,L_G)}{\left[\sum_{j=2}^{G} \sum_{i=1}^{j-1} L_j L_i+1\right]},\;\;\mbox{subject to} \;\;\sum_{i=1}^{G} L_i = K.
\end{equation}
As shown in Appendix \ref{app:optimal_grouping}, the solution of this discrete optimization problem is attained when the number of pilot transmissions for each group is equal to $1$, i.e., $L_i=1$ for any $i$ and $G=K$; note that the Avalanche approach is optimal in this sense.  In this case, the number of antennas that can be calibrated is $\frac{1}{2}G(G-1)+1$. Thus, for a given array size $M$, the number of channel uses grows only of the order of $\sqrt{M}$, which is faster than $\mathcal{O}{(M)}$ in Argos and the method of Rogalin et al.\footnote{The number of channel uses needed by the method in \cite{rogalin2014scalable} is $M$ if we perform round-robin broadcasting for each antenna assuming that the all channels between antennas are constant during the whole calibration process whereas it would be $\mathcal{O}{(M^2)}$ if we perform bi-directional transmission independently for each antenna pair. Please refer to Section~\ref{sec:non_coh_accum} for more details.} \cite{rogalin2014scalable}. 
Remark also that it is not necessary for the groups to be of equal size.

\section{Optimal estimation and performance limits}
\label{sec:crb}

In order to derive estimation error bounds for the reciprocity parameters, we should not exclude a priori any data obtained during the training phase, which is what we shall assume here.
In this section, we derive the CRB and associated ML estimation for the unified calibration scheme based on antenna partition. 
In order to obtain tractable results, we rely on a bilinear model to represent the calibration process.
From \eqref{eqn:sig_mod3}, we have
\begin{equation}
\begin{aligned}
 \Y_{i \rw j}  &=  \R_j \C_{i \rw j} \T_i \mathbf{P}_{i} +  \N_{i \rw j} \\
&=   \underbrace{\R_j \C_{i \rw j} \R_i^T}_{\bm{\mathcal{H}}_{i \rw j}} \F_{i} \mathbf{P}_{i} +   \N_{i \rw j}, \\
\end{aligned}
\end{equation}
where $ \F_i = \R_i^{-T}\T_i$ is the calibration matrix for group $i$.
We define $\bm{\mathcal{H}}_{i \rw j} = \R_j \C_{i \rw j} \R_i^T$ to be a auxiliary internal channel (not corresponding to any physically measurable quantity) that appears as a nuisance parameter in the estimation of the calibration parameters. 
Note that the auxiliary channel $\bm{\mathcal{H}}_{i \rw j}$ inherits the reciprocity from the OTA channel $\C_{i \rw j}$:
$\bm{\mathcal{H}}_{i \rw j} = \bm{\mathcal{H}}_{j \rw i}^T$.
Upon applying the vectorization operator for each bidirectional transmission between groups $i$ and $j$, we have, similarly to \eqref{eqn:lineq_calib_reduced}
\begin{equation}\label{Yij}
\mbox{vec}(\Y_{i \rw j})  =  (\mathbf{P}_i^{T} * \bm{\mathcal{H}}_{i \rw j})\, \mathbf{f}_{i} +   \mbox{vec}(\N_{i \rw j}).
\end{equation}
On the reverse direction, using $\bm{\mathcal{H}}_{i \rw j} = \bm{\mathcal{H}}_{j \rw i}^T$, we have
\begin{equation}\label{YjiT}
\mbox{vec}(\Y_{j \rw i}^T ) =  (\bm{\mathcal{H}}_{i \rw j}^{T} * \mathbf{P}_{j}^{T}) \f_j + \mbox{vec}(\N_{j \rw i})^T.
\end{equation}
Alternatively, \eqref{Yij} and \eqref{YjiT} may also be written as
\begin{equation}
\left\{
\begin{aligned}
\mbox{vec}(\Y_{i \rw j})  &= \left[(\F_i  \mathbf{P}_i)^T \otimes \I\right]\mbox{vec}(\bm{\mathcal{H}}_{i \rw j}) +   \mbox{vec}(\N_{i \rw j}) \\
\mbox{vec}(\Y_{j \rw i}^T)  &=\left[\I \otimes (\mathbf{P}_j^T\F_j)\right]\mbox{vec}(\bm{\mathcal{H}}_{i \rw j}) +   \mbox{vec}(\N_{j \rw i}).
\end{aligned}
\right.
\label{eqstackRxsignal}
\end{equation}
Stacking these observations into a vector
$\y = \left[  \mbox{vec}(\Y_{1 \rw 2})^{T} \, \mbox{vec}(\Y_{2 \rw 1}^T)^T \,  \mbox{vec}(\Y_{1 \rw 3})^{T}\dots \right]^T$, the above two alternative formulations can be summarized into
\begin{equation}
\label{eq:CRB_formulation}
\begin{aligned}
\y &=  \bm{\mathcal{H}}(\h,\mathbf{P})\f +  \n \\
              &=  \bm{\mathcal{F}}(\f,\mathbf{P})\h + \n ,
\end{aligned}
\end{equation}
where $\h = \left[\mbox{vec}(\bm{\mathcal{H}}_{1\rw 2})^T \, \mbox{vec}(\bm{\mathcal{H}}_{1\rw 3})^T \, \mbox{vec}(\bm{\mathcal{H}}_{2 \rw 3})^T \dots \right]^{T}$, and $\n$ is the corresponding noise vector. 
The composite matrices $\bm{\mathcal{H}}$ and $\bm{\mathcal{F}}$ are given by,
\begin{equation}
\label{eqn:crb_defns}
\begin{aligned}
 \bm{\mathcal{H}}(\mathbf{h},\mathbf{P}) &= \!\!\begin{bmatrix}
 \mathbf{P}_{1}^{T} * \bm{\mathcal{H}}_{1 \to 2} & 0 & 0 & \dots \\
 0 & \bm{\mathcal{H}}_{1 \to 2}^{T} * \mathbf{P}_{2}^{T} & 0 & \dots \\
 \mathbf{P}_{1}^{T} * \bm{\mathcal{H}}_{1 \to 3}\!\! & 0 & 0 & \dots \\
 0 & 0 & \bm{\mathcal{H}}_{1 \to 3}^{T} * \mathbf{P}_{3}^{T}  & \dots \\
\vdots &  \vdots &  \vdots &  \ddots 
\end{bmatrix} \\
\bm{\mathcal{F}}(\f,\mathbf{P}) &= \!\!\begin{bmatrix}
\mathbf{P}_{1}^{T}\mathbf{F}_{1}   \otimes \I & 0 & 0& 0 & \dots \\
\I \otimes   \mathbf{P}_{2}^{T}\F_{2}  & 0 & 0& 0 & \dots \\
 0 &  \mathbf{P}_{1}^{T} \F_{1}  \otimes \I & 0& 0  & \dots \\
0&  \I \otimes   \mathbf{P}_{3}^{T}\mathbf{F}_{3}  & 0& 0 & \dots \\
0& 0 &  \mathbf{P}_{2}^{T} \F_{2}  \otimes \I &  0  & \dots \\
0& 0&  \mathbf{I} \otimes   \mathbf{P}_{3}^{T}\F_{3}  &  0 & \dots \\
\vdots &  \vdots &  \vdots &  \vdots& \ddots 
\end{bmatrix}.
\end{aligned}
\end{equation}
The scenario is now identical to that encountered in some blind channel estimation scenarios and hence we can take advantage of some existing tools \cite{semi2000carvalho},\cite{carvalho2012PQML}, which we exploit next.

\subsection{Cram\'er-Rao bound}
\label{subsec:cramer_rao_bound}

Treating $\mathbf{h}$ and $\mathbf{f}$ as deterministic unknown parameters, and assuming that the receiver noise $\n$ is  distributed as  
$\mathcal{CN}(0,\sigma^2 \mathbf{I})$, the Fisher Information Matrix (FIM) $\J$ for jointly estimating $\f$ and $\h$ can immediately be obtained from
\eqref{eq:CRB_formulation} as 
\begin{equation}
\begin{aligned}
\J = \frac{1}{\sigma^2} \begin{bmatrix}
\bm{\mathcal{H}}^H \\
\bm{\mathcal{F}}^H
\end{bmatrix}
\begin{bmatrix}
\bm{\mathcal{H}} \quad
\bm{\mathcal{F}}
\end{bmatrix}.
\end{aligned}
\label{eqFIM}
\end{equation}
The computation of the CRB requires $\mathbf{J}$ to be non-singular. However, for the problem at hand, $\mathbf{J}$ is inherently singular. In fact, the calibration factors (and the auxiliary channel) can only be estimated up to a complex scale factor since the received data \eqref{eq:CRB_formulation}  involves the product of the channel and the calibration factors, 
$\bm{\mathcal{H}} \f = \bm{\mathcal{F}} \h$. As a result the FIM has the following null space \cite{carvalho2004ident},\cite{deCarvalho99blindCRB}
\begin{equation}
\begin{aligned}
\mathbf{J} \begin{bmatrix} \f \\
-\h \end{bmatrix} =  \frac{1}{\sigma^2} \begin{bmatrix}
\bm{\mathcal{H}} & \bm{\mathcal{F}}
\end{bmatrix}^{H}
\; 
(
\bm{\mathcal{H}} \f -
\bm{\mathcal{F}} \h
) = \mathbf{0}.
\end{aligned}
\end{equation}
To determine the CRB when the FIM is singular, constraints have to be added to regularize the estimation problem. 
As the calibration parameters are complex, one complex constraint corresponds to two real constraints. 
Another issue is that we are mainly interested in the CRB for $\f$, the parameters of interest, in the presence of the nuisance parameters $\h$.
Hence we are only interested in the $(1,1)$ block of the inverse of the $2\times 2$ block matrix $\mathbf{J}$ in \eqref{eqFIM}.
Incorporating the effect of the constraint \eqref{eqn:genconstr} on $\f$, we can derive from \cite{deCarvalho99blindCRB} the following constrained CRB for $\f$
\begin{equation}
\mbox{CRB}_{\f} = \sigma^2\cV_{\f} \left( \cV_{\f}^H\bmcH^{H}\bmcP_{\bmcF}^{\perp}\bmcH \cV_{\f}\right)^{-1}\cV_{\f}^H
\label{eqCRBf}
\end{equation}
where 
$\bmcP_{\X} = \X(\X^H\X)^{\dagger}\X^H$ and $\bmcP_{\X}^{\perp} = \I - \bmcP_{\X}$ are the projection operators on resp. the column space of matrix $\X$ and its orthogonal complement, and $\dagger$ corresponds to the Moore-Penrose pseudo inverse. Note that in some group calibration scenarios, $\bm{\mathcal{F}}^{H}\bm{\mathcal{F}}$ can be singular (i.e, $\h$ could be not identifiable even if $\f$ is identifiable or even known).
The $M\times (M\! -\! 1)$ matrix $\cV_{\f}$ is such that its column space spans the orthogonal complement of that of $\frac{\partial \cC(f)}{\partial \f^*}$, i.e.,
$
\bmcP_{\cV_{\f}} = \bmcP_{\frac{\partial \cC}{\partial \f^*}}^{\perp}
$.

It is shown in \cite{carvalho2004ident},\cite{deCarvalho99blindCRB},\cite{carvalho2000cramer} that a choice of constraints such that their linearized version 
$\frac{\partial \cC}{\partial \f^*}$
fills up the null space of the FIM results in the lowest CRB, while not adding information in subspaces where the data provides information.
One such choice is the set \eqref{eqRe}, \eqref{eqIm} (NPC). Another choice is \eqref{eqn:LinConstr} with $\g=\f$.
With such constraints, $\frac{\partial \cC}{\partial \f^*}\sim\f$ which spans the null space of $\bmcH^{H}\bmcP_{\bmcF}^{\perp}\bmcH$.
The  CRB then corresponds to the pseudo inverse of the FIM and \eqref{eqCRBf} becomes
\begin{equation}
\mbox{CRB}_{\f} = {\sigma^2}\left( \bmcH^{H}\bmcP_{\bmcF}^{\perp}\bmcH \right)^{\dagger}\, .
\end{equation}
If the FCC constraint is used instead (i.e., \eqref{eqn:LinConstr} with $\g=\eo_1$, $c=1$), the corresponding CRB is \eqref{eqCRBf} 
where $\cV_{\f}$ corresponds now to an identity matrix without the first column (and hence its column space is the orthogonal complement of that of $\mathbf{e}_1$). 

Note that \cite{vieira2017reciprocity} also addresses the CRB for a scenario where transmission happens one antenna at a time. 
The relative calibration factors are derived from the absolute Tx and Rx side calibration parameters, which become identifiable because a model is introduced for the internal propagation channel. 
In this Gaussian prior the mean is taken as the line of sight (LoS) component (distance induced delay and attenuation) and complex Gaussian non-LoS (NLOS) components are contributing to the covariance of this channel as a scaled identity matrix. 
The scale factor is taken 60dB below the mean channel power. This implies an almost deterministic prior for the (almost known) channel and would result in underestimation of the CRB, as noted in \cite[Sec.~III-E-2]{vieira2017reciprocity}. 

\subsection{Maximum likelihood estimation}

We now turn our focus to the design of an optimal estimator. From \eqref{eq:CRB_formulation} we get the negative log-likelihood up to an additive constant, as
\beq
\frac{1}{\sigma^2} \left\|\y  - \bm{\mathcal{H}}(\h,\mathbf{P})\f \right\|^2 = 
\frac{1}{\sigma^2} \left\|\y  - \bm{\mathcal{F}}(\f,\mathbf{P})\h \right\|^2 \; .
\label{eqML1}
\eeq
The maximum likelihood estimator of $(\h,\f)$, obtained by minimizing \eqref{eqML1}, can be computed using alternating optimization on $\h$ and $\f$, which leads to a sequence of quadratic problems. As a result, for given $\f$, we find
$\hh  = (\bm{\mathcal{F}}^{H}\bm{\mathcal{F}})^{-1} \bm{\mathcal{F}}^{H}\y$ and for given $\h$, we find
$\fh  = (\bm{\mathcal{H}}^{H}\bm{\mathcal{H}})^{-1} \bm{\mathcal{H}}^{H}\y$. This leads to the Alternating Maximum Likelihood (AML) algorithm (Algorithm \ref{alg_aml})  \cite{semi2000carvalho,carvalho2012PQML} which  iteratively maximizes the likelihood by alternating between the desired parameters $\f$ and the nuisance parameters $\h$ for the formulation \eqref{eq:CRB_formulation}\footnote{The method used in \cite{vieira2017reciprocity} to derive the ML estimator, although called ``Expectation Maximization'' in the original paper,  actually corresponds to the AML scheme, but using quadratic regularization terms for both $\f$ and $\h$ which can be interpreted as Gaussian priors and which may improve estimation in ill-conditioned cases.}.

\begin{algorithm}
\caption{Alternating maximum likelihood (AML)}
\label{alg_aml}
\begin{algorithmic}[1]
\STATE \textbf{Initialization:} Initialize $\fh$ using existing calibration methods (e.g. the method in \ref{subsec:LS}) or as a vector of all 1's.
\REPEAT
\STATE Construct $\bm{\mathcal{F}}$ as in \eqref{eqn:crb_defns} using $\fh$.\\
           $\hh     = (\bm{\mathcal{F}}^{H}\bm{\mathcal{F}})^{-1} \bm{\mathcal{F}}^{H}\tilde{\y}$ \\
\STATE Construct $\bm{\mathcal{H}}$ as in \eqref{eqn:crb_defns} using $\hh$.\\
           $\fh      =  (\bm{\mathcal{H}}^{H}\bm{\mathcal{H}})^{-1} \bm{\mathcal{H}}^{H} \tilde{\y}$
\UNTIL{the difference on the calculated $\hat{\f}$ between two iterations is small enough.}
\end{algorithmic}
\end{algorithm}

\subsection{Maximum likelihood vs. least squares}

At first, it would seem that the ML and CRB formulations above are unrelated to the LS method introduced in Section~\ref{sec:uni_rep} and used in most existing works. 
However, consider again the received signal in a pair $(i,j)$ as in \eqref{eqstackRxsignal}.
Eliminating the common auxiliary channel $\bmcH_{i \rw j}$, we get the elementary equation \eqref{eqn:obs_general} for the LS method \eqref{eq:eqn_basic_vec6} or \eqref{eq:eqn_basic_vec7}. 
Equivalently to \eqref{eqn:lineq_calib_reduced}, one obtains
\begin{equation}
\label{eqn:lineq_calib_reduced2}
\bmcY(\mathbf{P})\f = \bmcF^{\perp H}\y = \tin,
\end{equation}
where 
\begin{equation}
\label{eqn:Pperp}
\mbox{}\!\!
\bmcF^{\perp}\!=\! \begin{bmatrix}
\I\otimes (\F_2\bP_2)^* \!\!\!\! & 0 & 0& 0 & \dots \\
-(\F_1\bP_1)^*\otimes \I\!\!\!\!\!  & 0 & 0& 0 & \dots \\
 0 &  \I\otimes(\F_3\bP_3)^*\!\!\!\! & 0& 0  & \dots \\
0&  -(\F_1\bP_1)^*\otimes\I \!\!\!\!\! & 0& 0 & \dots \\
0& 0 &  \I\otimes(\F_3\bP_3)^*\!\!  &  0  & \dots \\
0& 0&  -(\F_2\bP_2)^*\otimes\I \!\!\! &  0 & \dots \\
\vdots &  \vdots &  \vdots &  \vdots& \ddots 
\end{bmatrix},
\end{equation}
such that the column space of $\bmcF^{\perp}$ corresponds to the orthogonal complement of the column space of $\bmcF$ (see Appendix~\ref{app:Fperp}) assuming that either $M_i\geq L_i$ or $L_i\geq M_i$ for all $1\leq i\leq G$.
Now, the ML criterion in \eqref{eqML1} is separable in  $\f$ and $\h$. Optimizing 
\eqref{eqML1} w.r.t. $\h$ leads to $\h  = (\bm{\mathcal{F}}^{H}\bm{\mathcal{F}})^{\dagger} \bm{\mathcal{F}}^{H}\y$ 
as mentioned earlier. Substituting this estimate for $\h$ into \eqref{eqML1} yields a ML estimator $\hat{\f}$ minimizing
\beq
\y^H \bmcP_{\bmcF}^{\perp}\y = \y^H \bmcP_{\bmcF^{\perp}}\y = \y^H \bmcF^{\perp}(\bmcF^{\perp H}\bmcF^{\perp})^{\dagger}\bmcF^{\perp H}\y,
\label{eqML-LS1}
\eeq
where we used $\bmcP_{\bmcF}^{\perp} = \bmcP_{\bmcF^{\perp}}$.
This should be compared to the least-squares method which consists in minimizing $\|\bmcF^{\perp H}\y\|^2 = \| \bmcY\f\|^2$ 
in \eqref{eq:eqn_basic_vec6} or \eqref{eq:eqn_basic_vec7}.
Hence \eqref{eqML-LS1} can be interpreted as an optimally weighted least-squares method since from \eqref{eq:CRB_formulation}
$\bmcF^{\perp H}\y = \bmcF^{\perp H}\n = \tin$ leads to colored noise with covariance matrix
$ \sigma^2 \bmcF^{\perp H}\bmcF^{\perp}$.
The compressed log-likelihood in \eqref{eqML-LS1} can now be optimized using a variety of iterative techniques such as
Iterative Quadratic ML (IQML), Denoised IQML (DIQML) or Pseudo-Quadratic ML (PQML) \cite{carvalho2012PQML},
and initialized with the least-squares method.
It is not clear though whether accounting for the optimal weighting in ML would lead to significant gains in performance.
The weighting matrix (before inversion) $ \bmcF^{\perp H}\bmcF^{\perp}$ is block diagonal with a square block corresponding to
the pair of antenna groups $(i,j)$ being of dimension $L_iL_j$. 
If all $L_i=1$, then $ \bmcF^{\perp H}\bmcF^{\perp}$ is a diagonal matrix. If furthermore all $M_i=1$ (groups of isolated antennas),
all pilots are of equal magnitude, and if all calibration factors would be of equal magnitude, then 
$ \bmcF^{\perp H}\bmcF^{\perp}$ would be just a multipe of identity and hence would not represent any weighting.
We shall leave this topic for further exploration.
In any case, the fact that the CRB derived above and the ML and LS methods are all based on the signal model
\eqref{eq:CRB_formulation} shows that, 
the CRB above is the appropriate CRB for the estimation methods discussed here.

\subsection{Calibration bias at low SNR}
\label{sec_calibration_bias}

Whereas the CRB applies to unbiased estimators, at low SNR the estimators are biased which turns out to lead to mean square error (MSE) saturation.
In the case of a norm constraint, $\|\fh\|^2 = \|\f\|^2$, due to the triangle inequality
\beq
\|\fh-\f\| \leq \|\fh\| + \|\f\| = 2\|\f\|,
\eeq
MSE $ = \mathbb{E}[\|\fh-\f\|^2] \leq 4 \|\f\|^2$. 
However, MSE saturation occurs also in the case of a linear constraint.
We shall provide here only some brief arguments.
For a linear constraint of the form \eqref{eqn:LinConstr}, 
the least-squares method leads to \eqref{eq:eqn_basic_vec6}.
As the SNR decreases, the noise part $\bmcN$ of $\bmcY$ will eventually dominate $\bmcY$.
Hence $ \fh = \frac{c}{\g^H (\bmcN^H\bmcN)^{-1}\g} (\bmcN^H\bmcN)^{-1}\g$ in which the coefficients
as LS estimation coefficients will tend to be bounded.
To take a short-cut, consider replacing $\bmcN^H\bmcN$ by its mean $\mathbb{E}[\bmcN^H\bmcN] = c'\, \I$. Then we get 
$ \fh = \frac{c}{\g^H\g} \g$ which is clearly bounded. Hence $\fh$ will be strongly biased with bounded MSE.



\section{Non-coherent accumulation}
\label{sec:non_coh_accum}

\subsection{Overview}

We have assumed in Sections \ref{sec:uni_rep} and \ref{sec:prior} that the channel is constant during the whole calibration process, which may become questionable if the number of antennas becomes very large since more time is then needed to accomplish the whole calibration process. As a consequence, it is possible that we cannot accumulate enough observations within a single channel coherence time and frequency block. In this section, we consider such calibration algorithms, which can jointly use data accumulated during several independent fades of the OTA channel; since the requirement to calibrate during a single coherence interval of the channel is lifted, we denote this by \emph{non-coherent} accumulation of calibration data. Such approaches are essential for the calibration of massive MIMO systems.

Let us consider the method of Rogalin et al. as an example. If the channel is constant during the whole calibration process, one can readily use the (coherent) method detailed in Section~\ref{subsec:LS}, broadcasting pilots from each antenna in a round-robin manner when all other antennas are listening, thus $M$ slots are needed to accomplish the whole process. On the other hand, if the coherence time is not large enough, a non-coherent way to accumulate observations can be performing bi-directional transmissions for each antenna pair independently (in this case, we only require that the forward and backward transmissions are performed during the same coherence slot for each antenna pair); this requires therefore $M(M-1)$ slots. Here, we see that the non-coherent accumulation is enabled at the cost of spending more resources on calibration ($M(M-1)$ transmissions vs. $M$ transmissions for the coherent case).
Some papers also implicitely use non-coherent accumulations; see for example \cite{jiang2016novel} who derives a Total Least-Squares (TLS) estimator from such measurements.

Let us extend the signal model in Section~\ref{sec:uni_rep} by allowing to accumulate measurements over several time slots beyond the channel coherence time (the channel can only be assumed constant within each slot, not necessarily across the slots). We assume that these are indexed by $1\leq t\leq T$, so that $T$ represents the number of coherent slots at disposal. Clearly, the OTA reciprocity equation $\C_{i\rw j} = \C_{j\rw i}^T$ holds only for measurements obtained during the same time slot. However, measurements related to several groups of antennas obtained during multiple non-coherent time slots can be successfully combined to perform joint calibration of the complete array, as shown next.
Let us assume that, during a given coherent slot $t$, a subset $\mathcal{G}(t)$ of the groups forming the partition of the array transmit training signals; we require that $\mathcal{G}(t)$ has at least two elements. When group $A_i$, $i \in \mathcal{G}(t)$ is transmitting, the received signal at group $A_j$, $j \in \mathcal{G}(t)$, $j\neq i$ is written as $\Y_{j\rw i,t} = \R_j\C_{i\rw j, t}\T_i\mathbf{{P}}_{i,t} + \N_{j,t}$, and $\Y_{i\rw j,t}$ is defined similarly.
\eqref{eqn:obs_general2} then becomes
\begin{equation}
\label{eqn:obs_non_coh}
(\Y_{j\rw i,t}^T*\mathbf{P}_{i,t}^T)\f_i - (\mathbf{P}_{j,t}^T*\Y_{i\rw j,t}^T)\f_j = \tin_{ij,t}.
\end{equation}
Stacking these equations similarly to \eqref{eqn:lineq_calib_reduced}, but with respect to the $i,j \in \mathcal{G}(t)$, gives $\bmcY_t(\mathbf{P}_t)\f =\tin_t$ for each time slot $t$.

\subsection{LS estimation}

The LS estimator of the calibration matrix is thus, taking into account all observations accumulated over the $T$ slots,
\beq
\label{eqn:non_coh_acc}
\bal
 \hat{\f}  &= \arg\min_{\f}\sum_{t=1}^{T} \!\sum_{\stackrel{i,j \in \mathcal{G}(t)}{i\neq j}}\!\left\|(\Y_{j\rw i,t}^T\! *\!\mathbf{P}_{i,t}^T)\f_i - (\mathbf{P}_{j,t}^T\! *\!\Y_{i\rw j,t}^T)\f_j \right\|^2\\
&= \arg\min_{\f}\|\bmcY(\mathbf{P})\f \|^2,
\eal
\eeq
where the minimum is taken either under the constraint $f_1=1$ or $\|\f\|=1$ and $\bmcY(\mathbf{P})=[\bmcY_1(\mathbf{P}_1)^T,\dots,\bmcY_T(\mathbf{P}_T)^T]^T$. Therefore, the approach of \eqref{eqn:non_coh_acc} is very similar to \eqref{eq:eqn_basic_vec6} and \eqref{eq:eqn_basic_vec7}. This shows that calibration using a joint estimator based on non-coherent measurements can be readily implemented by making sure that the measurements $\Y_{j\rw i,t}$ and $\Y_{i\rw j,t}$ appearing in each term of the sum above have been obtained during the same coherence interval. Note also that this approach can allow to collect multiple measurements across independent channel fades between the same pair $(i,j)$ of antenna groups, hence providing a way to increase the accuracy (by averaging over multiple noise realizations) and robustness (by minimizing the effect of a single catastrophic realization of the internal channel which could yield a rank-deficient set of linear equations for a given $t$) of the estimator.

\subsection{Optimal grouping}

Statements similar to those in Section~\ref{sec:fast_calib} can be made for non-coherent group-based fast calibration. The maximization proposed in Section~\ref{sec:fast_calib} is still valid in this context leading to an optimal number of groups equal to the number of coherent slots $G=K$. Therefore, since $\frac{1}{2}K(K-1)$ independent rows in $\bmcY(\mathbf{P})$ are accumulated per coherent slot, if we fix the number of antennas to $M$, the number of coherent slots $T$ should satisfy $\frac{T}{2}K(K-1)\geq M-1$ in order to calibrate all antenna elements. Note that the total number of calibrated antennas, equal to $\frac{T}{2}K(K-1)+1$, is linear in $T$ and quadratic in $K$, which confirms that it is more valuable to perform coherent measurements in order to speed up the calibration process. However, non-coherent accumulations allow to perform measurements sparsely in time. Such a calibration process can be interleaved with the normal data transmission or reception, leading to vanishing resource overhead.

\section{Numerical Validation}
\label{sec:simu_result}

In this section, we assess numerically the performance of various calibration algorithms and also compare them against their CRBs.
We first evaluate the proposed group-based fast calibration method from Section~\ref{sec:fast_calib}. 
We use $\mbox{MSE} = \mathbb{E}[\|\fh-\f\|^2]$ as the performance evaluation metric and the CRB as benchmark.
The Tx and Rx calibration parameters for the BS antennas are assumed to have random phases uniformly distributed over 
$[-\pi,\pi]$ and amplitudes uniformly distributed in the range $[1-\delta,1+\delta]$ where $\delta=0.1$. Except for the first coefficients which are fixed to 1 so that $f_1=1$ for the true $\f$.
In this way, regardless of whether the FCC or the NPC 
(i.e.\ \eqref{eqRe},\eqref{eqIm} with $c=||\f||^2$)
constraints are used, direct comparison of $\fh$ to $\f$ is possible for the MSE computation (in which the expectation is replaced by sample averaging).
For a fair comparison across different schemes, the number of channel uses should be the same. 
Hence, we compare the fast calibration method of Section~\ref{sec:fast_calib} against the Avalanche scheme proposed in \cite{papadopoulos2014avalanche}. 
Note that the Argos and Rogalin methods are not fast algorithms as they need channel uses of the order of $M$, so they cannot be compared with the fast calibration methods.
The number of antennas that transmit at each time instant (i.e. the group sizes of the 12 antenna groups) is shown in Table \ref{tab:num_ant_txn_per_chan_use}. FC-I corresponds to a fast calibration scheme where the antenna grouping is exactly the same as that of Avalanche. However, we also try a more equally partitioned grouping of antennas in FC-II. The pilots used for transmission have unit magnitudes with uniform random phases in $[-\pi,\pi]$. 
The channels between all the BS antennas are assumed to be i.i.d. Rayleigh fading.

\begin{table}[h!]
  \centering
  \caption{Number of antennas transmitting at each channel use for two Fast Calibration schemes.}
  \label{tab:num_ant_txn_per_chan_use}
  \begin{tabular}{|l|c|c|c|c|c|c|c|c|c|c|c|c|}
  \hline
  Scheme &  \multicolumn{12}{|c|}{Antennas transmitting per channel use. $M=64$} \\
  \hline
  FC-I          & 1 & 1 & 2 & 3 & 4 & 5 & 6 & 7 & 8 & 9 & 10 & 8\\
  \hline
  FC- II       & 5 & 5 & 5 & 5 & 5 & 5 & 5 & 5 & 6 & 6 & 6 & 6\\
  \hline
  Scheme &  \multicolumn{12}{|c|}{Antennas transmitting per channel use. $M=67$} \\
  \hline
  FC-I          & 1 & 1 & 2 & 3 & 4 & 5 & 6 & 7 & 8 & 9 & 10 & 11\\
  \hline
  FC- II       & 5 & 5 & 5 & 5 & 5 & 6 & 6 & 6 & 6 & 6 & 6 & 6\\
  \hline
  \end{tabular}
\end{table}

The performance of these schemes is depicted in Fig.~\ref{fig:MSR_fast_vs_avalanche} for $M=64$. 
From Section~\ref{sec:fast_calib}, it can be seen that the minimal number of channel uses required for calibration is $G=12 = \lceil\sqrt{2M}\rceil$. The performance is averaged over 500 realizations of channel and calibration parameters. 
Note that the Avalanche algorithm inherently uses the FCC in its estimation process. For comparison to methods using NPC, the Avalanche estimate $\fh$ is then rescaled in order to satisfy the NPC constraint.

As the CRB depends on the constraint used for calibration estimation, the corresponding CRBs for these approaches are also shown. However, note that the CRB  for the FC-I grouping applies to both the Avalanche method and the proposed fast calibration method (which performs least-squares \eqref{eqn:general_LS_min} over all available data jointly).  For each type of constraint, there are thus 3 MSE curves (Avalanche, FC-I and FC-II) and 2 CRB curves (for FC-I and FC-II). As the MSE curve is averaged over multiple channel realizations, the CRB plotted here is also an average over the CRB values corresponding to these channel realizations. 

In Fig.~\ref{fig:MSR_fast_vs_avalanche}, the performance of the proposed fast calibration with the FC-I grouping outperforms that of the Avalanche scheme. With $M=64$ and $G=12$ antenna groups, the overall system of equations is overdetermined: from \eqref{eqident} with $L_i=1$, $66=\frac{1}{2} G(G-1) > M-1 = 63$. This means that the proposed fast calibration, which  exploits all data jointly for the parameter estimation, has an advantage over the Avalanche method which solves exactly determined subsets of equations and hence suffers from error propagation. Also, the performance improves when the group sizes are allocated more equitably as in grouping scheme FC-II. Intuitively, the overall estimation performance of the fast calibration would be limited by the (condition number of the) largest group size and hence it is reasonable to use a grouping scheme that tries to minimize the size of the largest antenna group. These observations hold irrespective of the constraints used. Avalanche with the FCC constraint exhibits a huge MSE and hence most portions of this curve fall outside the range of Fig.~\ref{fig:MSR_fast_vs_avalanche}. Note also that the MSE in some cases falls below the CRB, see for instance the MSE NPC FC-I curve at low SNRs. This is because in this SNR region the MSE saturates due to bias and the CRB is no longer applicable as explained in Section~\ref{sec_calibration_bias}.
 
It is also illustrative to consider the case of $M=67$ antennas, which is the maximum number of antennas that can be calibrated with $G=12$ channel uses. As shown in section \ref{sec:fast_calib}, the best strategy is to divide the antennas into $G=12$ groups and letting each group transmit exactly once ($L_i=1$). This then results in a linear system of 66 equations \eqref{eqn:lineq_calib_reduced} plus one constraint in 67 unknowns. 
Indeed, \eqref{eqident} yields $66=\frac{1}{2} G(G-1) = M-1 = 66$.
Thus, the system of equations  is exactly determined by using an appropriate constraint to resolve the scale factor ambiguity. Hence, the error attained by any LS solution would be zero and the different constraints used for estimation would only lead to different scale factors in the calibration parameter estimates. So, all the solutions would be equivalent. Also, FC-I grouping leads to a block triangular structure with square diagonal blocks for the matrix $\bmcY$ defined in \eqref{eqn:cursivey} after removing the first column. Hence, the back substitution based solution performed by Avalanche is indeed the overall LS solution with the first coefficient known constraint. 
Thus, in Fig.~\ref{fig:MSR_fast_vs_avalanche_Nt67} where the performance of these schemes is compared for $M=67$, we see that the curves for Avalanche and fast calibration with the FC-I grouping overlap completely. In general, this behavior would occur whenever the number of antennas corresponds to the maximum that can be calibrated with the number of channel uses (see Sec.~\ref{sec:identifiability}), and the antenna grouping is similar to that for FC-I. At the range of SNRs considered, the MSE is saturated and is hence far below the CRB for this grouping. In fact, only a part of the CRB for the FC-I grouping can be seen as the rest of the curve falls outside the range of the figure. Indeed, though not shown in Fig.~\ref{fig:MSR_fast_vs_avalanche_Nt67}, the MSE curve with the FC-I grouping only starts to overlap with the corresponding CRB curve for SNR beyond 100dB!  However, it is important to note that the performance improves dramatically with a more equitable grouping of the antennas as can be seen from the curves for the FC-II grouping in the same figure.

\begin{figure}[!t]
\centering
\includegraphics[width=\columnwidth]{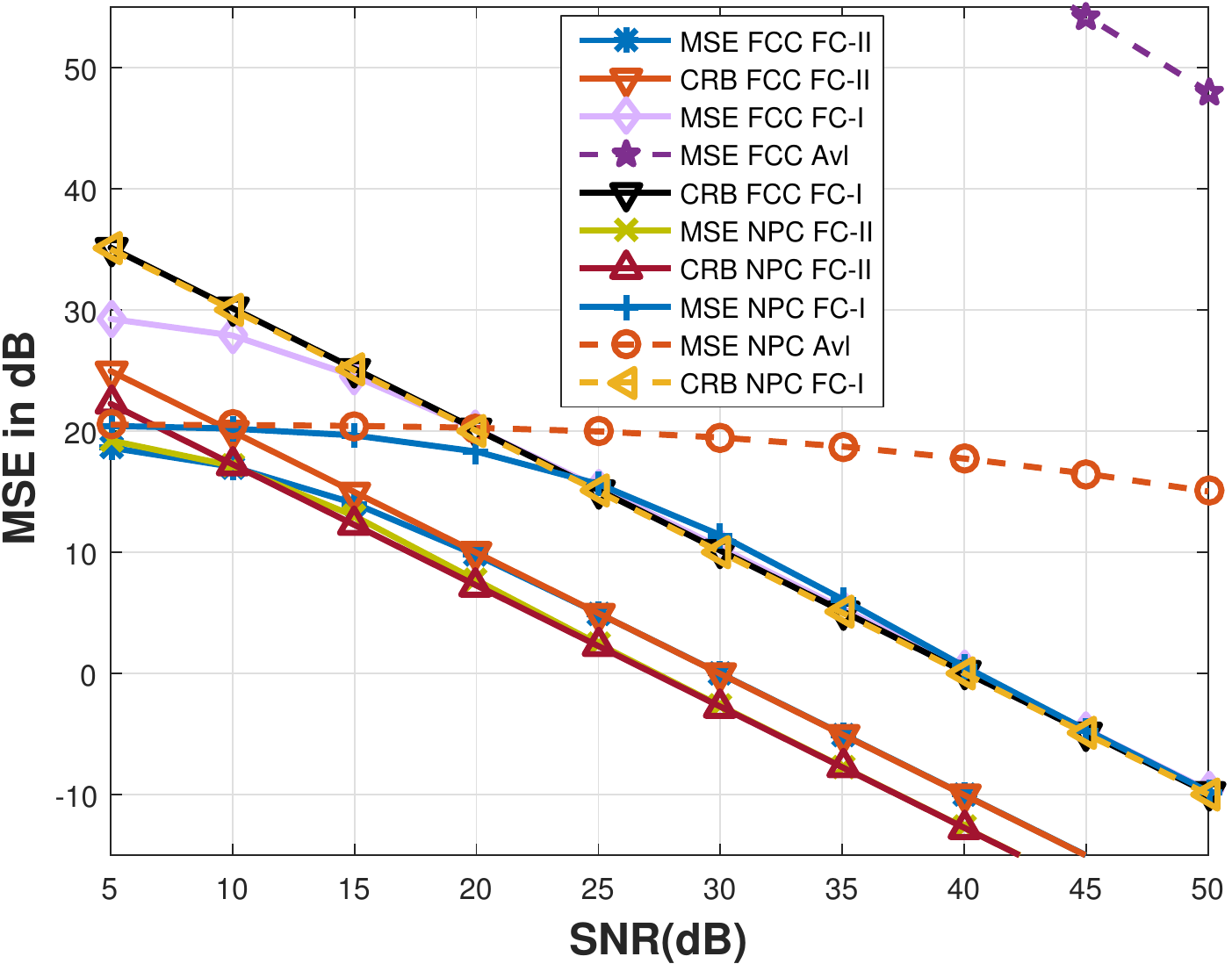}
\caption{Comparison of fast calibration with Avalanche scheme ($M=64$ and the number of channel use is 12). The curves are averaged across 500 channel realizations. The performance with both the FCC and NPC constraints is shown.}
\label{fig:MSR_fast_vs_avalanche}
\end{figure}

\begin{figure}[!t]
\centering
\includegraphics[width=\columnwidth]{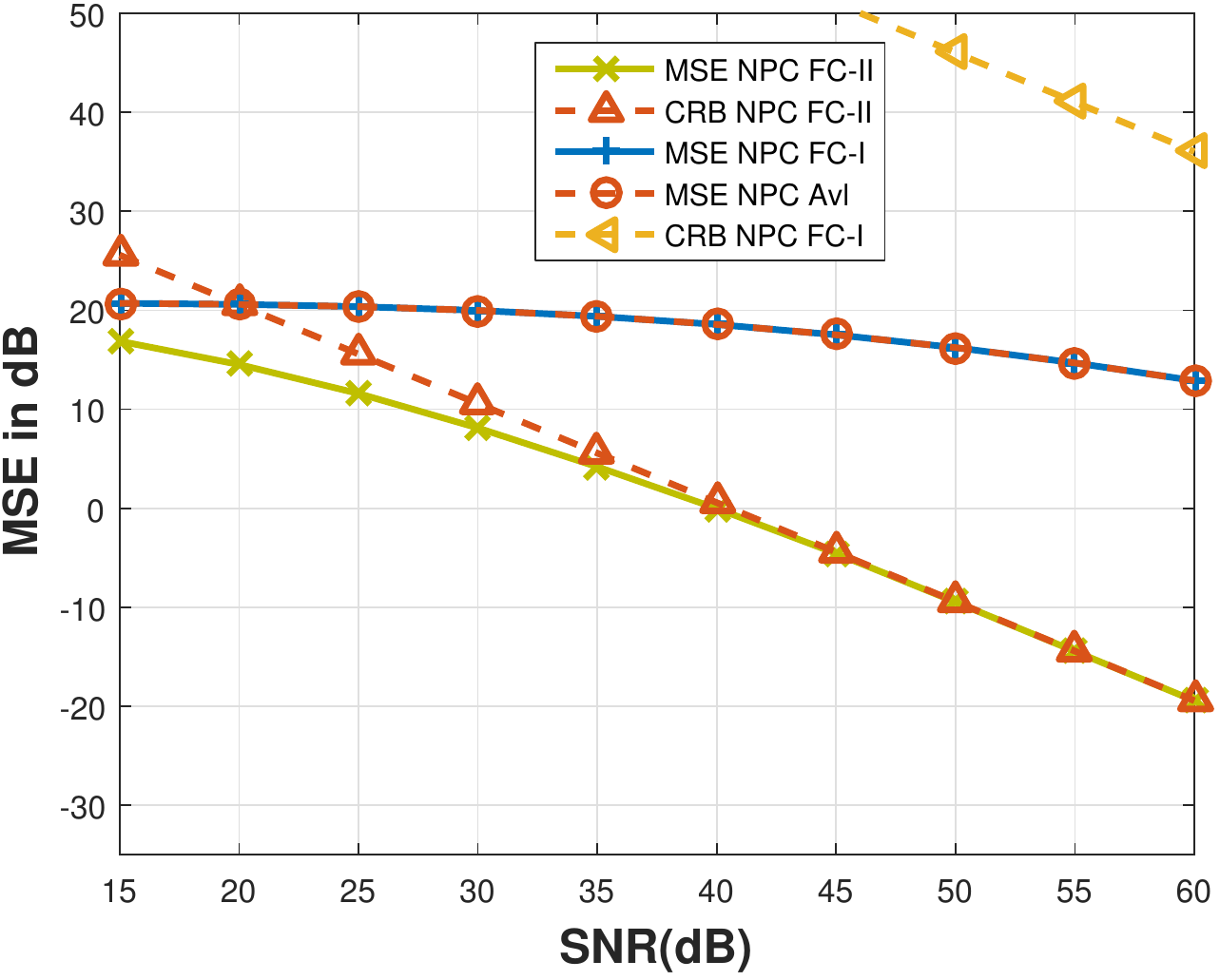}
\caption{Comparison of fast calibration with Avalanche scheme for $M=67$ and number of channel uses 12. The curves are averaged across 500 channel realizations. The NPC constraint is used for the MSE computation.}
\label{fig:MSR_fast_vs_avalanche_Nt67}
\end{figure}

In Fig.~\ref{fig:MSE_ArgRogAML_vs_crlb}, we consider slower transmit schemes that transmit from one antenna at a time ($G=M$) and compare their MSE performance with the CRB. The MSE with FCC for Argos, Rogalin \cite{rogalin2014scalable} and the  AML method in Algorithm \ref{alg_aml} is plotted. As expected, the Rogalin method improves over Argos by using all the bi-directional received data. AML outperforms the Rogalin performance at low SNR. These curves are compared with the CRB derived in \ref{subsec:cramer_rao_bound} for the FCC case and it can be seen that the AML curve overlaps with the CRB at higher SNRs. Also plotted is the CRB as given in \cite{vieira2017reciprocity}  assuming the internal propagation channel is fully known (the mean is known and the variance is negligible) and the  underestimation of the MSE can be observed as expected. To bring out the difference between the two CRB derivations, the amplitude variation parameter $\delta$ is chosen to be 0.5 to increase the range of values of Tx and Rx calibration parameters.

\begin{figure}[!t]
\centering
\includegraphics[width=\columnwidth]{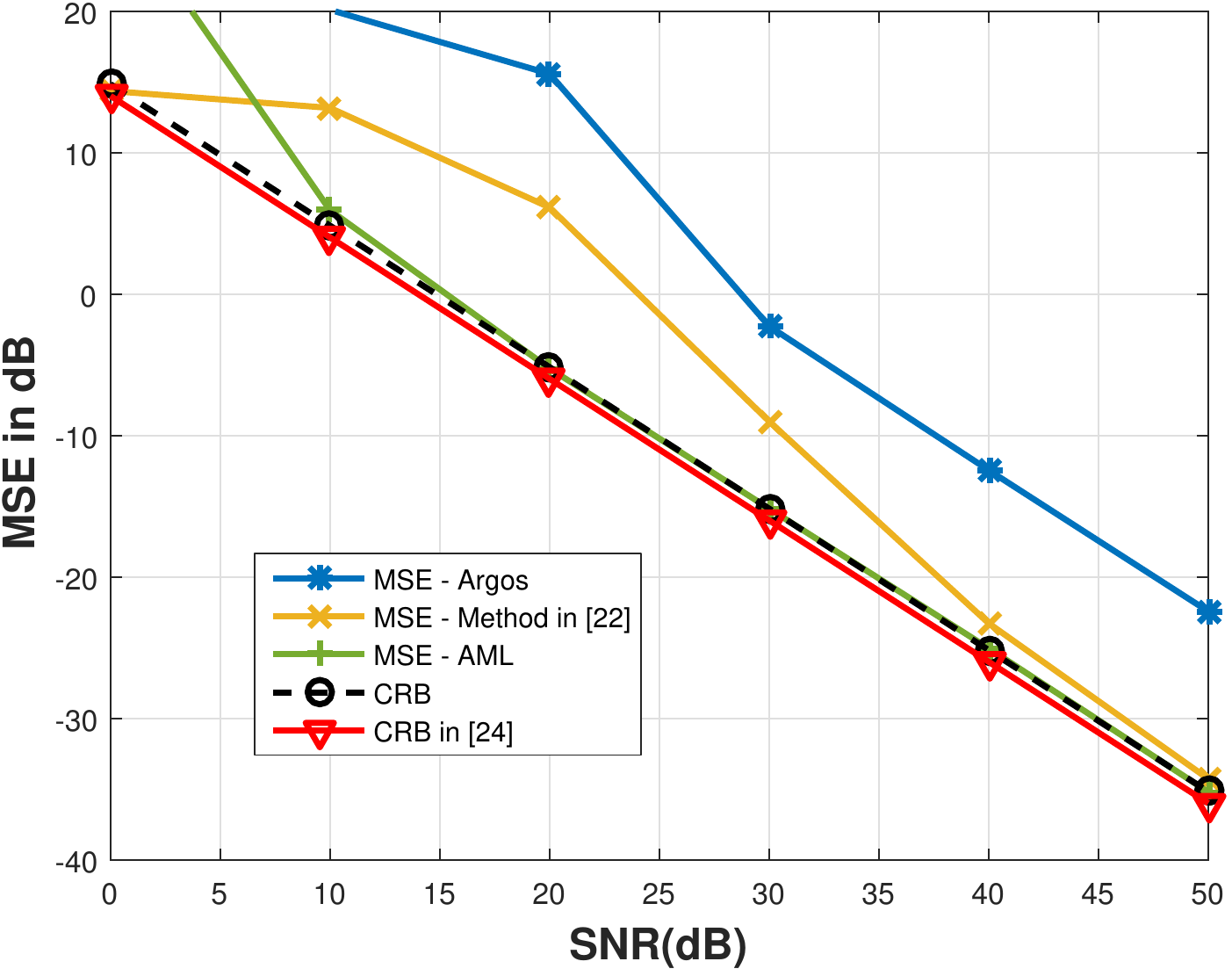}
\caption{Comparison of single antenna transmit schemes with the CRB ($G=M=16$). The curves are generated over one realization of an i.i.d.\  Rayleigh channel and known first coefficient constraint is used.}
\label{fig:MSE_ArgRogAML_vs_crlb}
\end{figure}

So far, we have focused on an i.i.d.\  intra-array channel model and we have seen in Fig.~\ref{fig:MSR_fast_vs_avalanche} and Fig.~\ref{fig:MSR_fast_vs_avalanche_Nt67} that the size of the transmission groups is an important parameter that impacts the MSE of the calibration parameter estimates.
We now consider a more realistic scenario where the intra-array channel is based on the geometry of the BS antenna array and make some observations on the choice of the antennas to form a group. We consider an array of $M=64$ antennas arranged as in Fig.~\ref{fig:BS_ant_pattern}. The path loss $(4 \pi \frac{d_{i \to j}}{\lambda})^2$ between any two antennas $i$ and $j$ is a function of their distance $d_{i \to j}$, and $\lambda$ is the wavelength of the received signal. In the simulations, the distance between adjacent antennas, $d$, is chosen as $\frac{\lambda}{2}$. The phase of the channel between any two antennas is modeled to be a uniform random variable in $[-\pi,\pi]$. Such a model was also observed experimentally in \cite{vieira2017reciprocity}. The SNR is defined as the signal to noise ratio observed at the receive antenna nearest to the transmitter.

\begin{figure}[!t]
\centering
\includegraphics[width=\columnwidth]{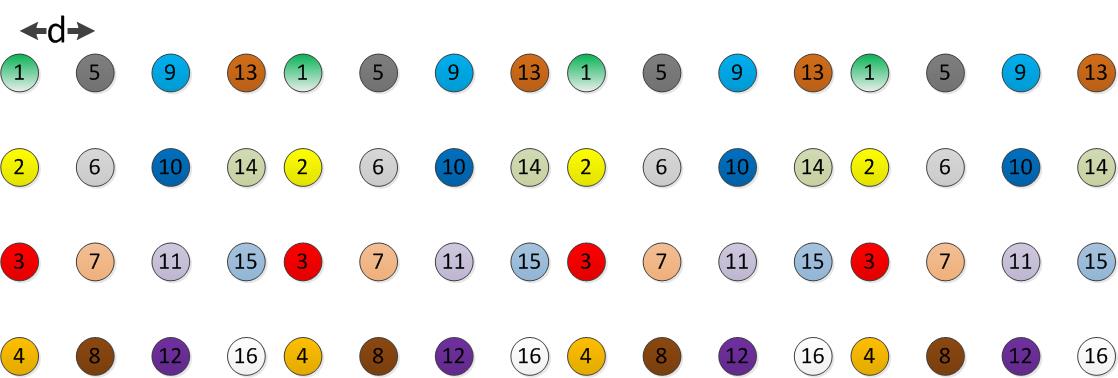}
\caption{64 antennas arranged as a 4 $\times$ 16 grid.}
\label{fig:BS_ant_pattern}
\end{figure}

Continuing with the same internal channel model, consider a scenario in which antennas transmit in $G=16$ groups of 4 each. Note that this is not the fastest grouping possible, but the example is used for the sake of illustration. We consider two different choices to form the antenna groups: 1) interleaved grouping corresponding to selecting antennas with the same numbers into one group as in  Fig.~\ref{fig:BS_ant_pattern}, 2) non-interleaved grouping corresponding to selecting antennas in each column as a group.
Fig.~\ref{fig:N_g16N_t64force_pil0N_chan1N_iter50selectchan_mode0cmp_mode1compute_MSE2LS_method2FFCalib0rnd100_crb_intlv} shows that interleaving of the antennas results in performance gains of about 10dB. Intuitively, the interleaving of the antennas ensures that the channel from a group to the rest of the antennas is as well conditioned as possible.
This example clearly shows that in addition to the size of the antenna groups, the choice of the antennas that go into each group also has a significant impact on the estimation quality of the calibration parameters.
\begin{figure}[!t]
\centering
\includegraphics[width=\columnwidth]
{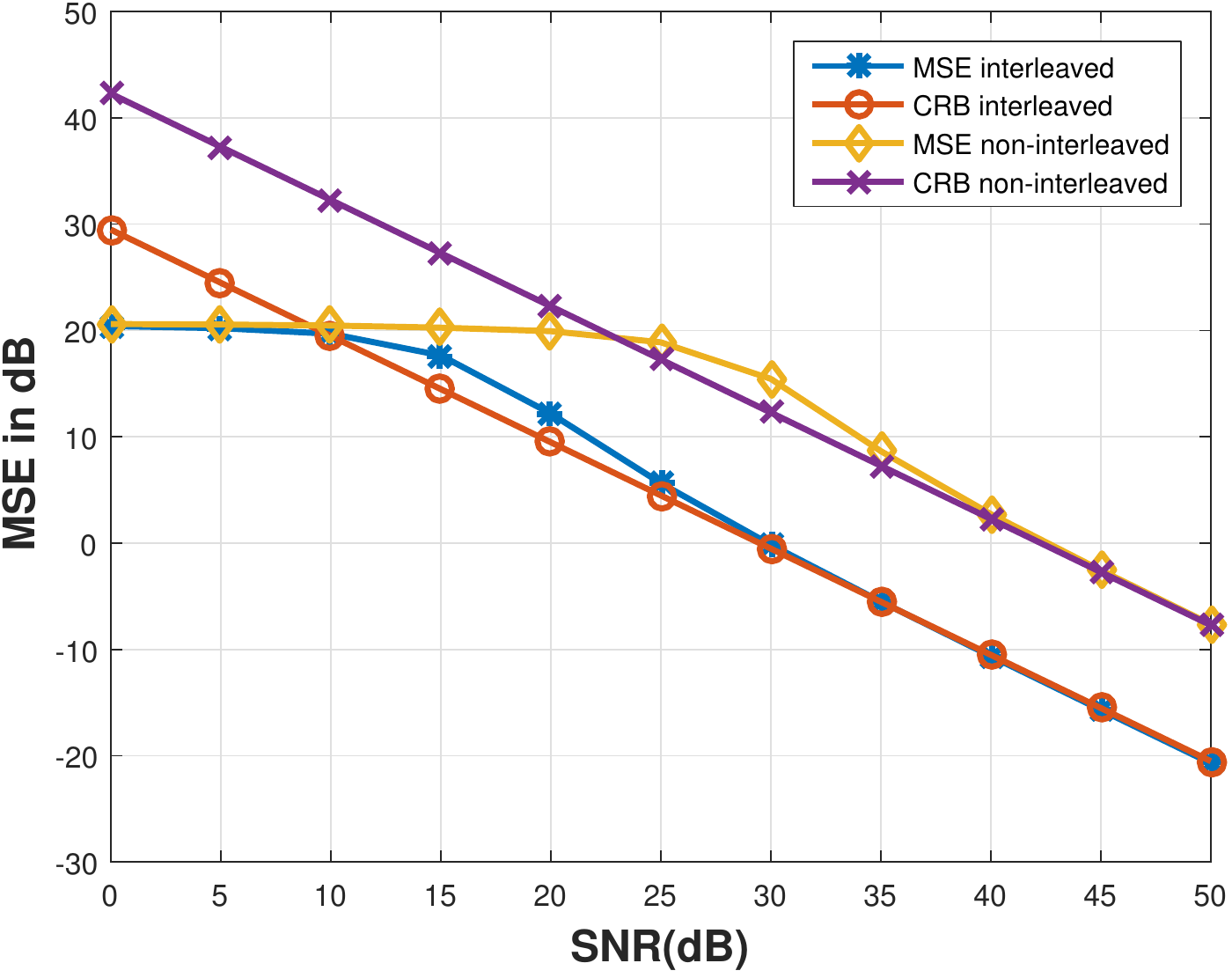}
\caption{Interleaved and non-interleaved MSE and CRB with NPC for an antenna transmit group size of 4 ($M=64$ and the number of channel uses is $G=16$).}
\label{fig:N_g16N_t64force_pil0N_chan1N_iter50selectchan_mode0cmp_mode1compute_MSE2LS_method2FFCalib0rnd100_crb_intlv}
\end{figure}

\section{Conclusions}
\label{sec:conclusions}

In this work we presented an OTA calibration framework which unifies existing calibration schemes. 
This framework opens up new calibration possibilities. As an example, we proposed a family of fast calibration schemes based on antenna grouping. The number of channel uses needed for the whole calibration process is of the order of the square root of the antenna array size rather than scaling linearly. In fact it can be as fast as the existing Avalanche calibration method, but avoiding the severe error propagation problem, thus greatly outperforming the Avalanche method, as has been shown by simulation results. 

We also presented a simple CRB formulation for the estimation of the relative calibration parameters. As the group calibration formulation encompasses the existing calibration methods, the CRB computation can be used to evaluate existing state of the art calibration methods as well. We then proposed a ML estimator and reveal the relationship between ML and LS estimation.

Moreover, we differentiated the notions of coherent and non-coherent accumulations for calibration observations. We illustrated that it is possible to perform calibration measurements using time slots that can be sparsely distributed over a relatively long time. This makes the calibration process consume a vanishing fraction of channel use resources and allows to minimize the impact on the ongoing data service. 

As illustrated by simulations, for the fast calibration, interleaved grouping has a better performance than non-interleaved grouping. However, the best antenna group definitions for given antenna group sizes is still an open question. 
Additionally, the optimal pilot design for calibration is unknown, which is an interesting topic for future work.

\appendices
\section{Optimal grouping}
\label{app:optimal_grouping}
\begin{lemma}
Fix $K\geq 1$. Let us define an optimal grouping as the solution $G^*, L_1^*,\dots,L_{G^*}^*$ of
\begin{equation}
\label{eq:discrete_optimization}
\max_{\sum_{i=1}^G L_i=K} \sum_{i<j} L_iL_j,
\end{equation}
then the optimal grouping corresponds to the case $L_1^*=\dots =L_{G^*}^*=1$ with $G^*=K$. The number of calibrated antennas is then equal to $\frac{1}{2}K(K-1)+1$.
\end{lemma}

\begin{proof}
Since the variables $L_1,\dots,L_G,G$ are discrete and $\sum_{i<j} L_jL_i$ is upper bounded by $K^2$, (\ref{eq:discrete_optimization}) admits at least one solution. Let $\Lbold = (L_1,\dots,L_G)$ be such a solution. We reason by contradiction: suppose that there exists $j$ such that $L_j>1$. Without loss of generality, we can suppose that $L_G>1$. Then, we can break up group $G$ and add one group which contains a single antenna, i.e., let us consider $\Lbold' = (L_1,\dots,L_G-1,1)$. In that case, it holds $\sum_{i=1}^G L_i = \sum_{i=1}^{G+1} L'_i=K$ and 
\begin{eqnarray*}
&&\sum_{j=2}^{G+1}\sum_{i=1}^{j-1} L'_jL'_i\\
&=&  \sum_{j=2}^{G-1}\sum_{i=1}^{j-1} L'_jL'_i + (L'_G+L'_{G+1})\sum_{i=1}^{G-1} L'_jL'_i + L'_GL'_{G+1}\\
 &=& \sum_{j=2}^{G}\sum_{i=1}^{j-1} L_jL_i+ L'_G> \sum_{j=2}^{G}\sum_{i=1}^{j-1} L_jL_i
\end{eqnarray*}
which contradicts the fact that $\Lbold$ is solution to (\ref{eq:discrete_optimization}). We conclude therefore that $L_j=1$ for any $j$ and $G^*=K$.
\end{proof}

\section{Construction of $\bmcF^{\perp}$}
\label{app:Fperp}
We show in the following that the column space of $\bmcF^{\perp}$ defined by \eqref{eqn:Pperp} spans the orthogonal complement of the column space of $\bmcF$ assuming that $\mathbf{P}_i$ is full rank for all $i$ and that either $L_i\geq M_i$ or $M_i\geq L_i$ for all i.
\begin{proof}
First, using $(\A\otimes\B)(\C\otimes\D) = (\A\C\otimes\B\D)$, it holds
\beq
\underbrace{\left[\I_{L_i} \otimes \mathbf{P}_j^T\F_j\;\; -   \mathbf{P}_i^T\F_i \otimes \I_{L_j}\right]}_{L_iL_j \times (L_iM_j + L_jM_i)}
\underbrace{\left[\begin{array}{l}
\mathbf{P}_i^T\F_i \otimes \I_{M_j} \\
\I_{M_i} \otimes \mathbf{P}_j^T\F_j
\end{array}\right]}_{(L_iM_j + L_jM_i) \times M_iM_j } = \0\; .
\label{eqMLLS2}
\eeq
Then, the row space of the left matrix of \eqref{eqMLLS2} is orthogonal to the column space of the right matrix. As $\bmcF$ in \eqref{eqn:crb_defns} and $\bmcF^{\perp\, H}$  are block diagonal with blocks of the form of \eqref{eqMLLS2}, it suffices then to prove that the following matrix $\M$ has full column rank, i.e.,  $L_iM_j + L_jM_i$,
which is then also its row rank
\begin{equation}
\M:=\left(\begin{array}{cc}
\I_{L_i}\otimes \mathbf{P}_j^T\F_j & -\mathbf{P}_i^T\F_i \otimes \I_{L_j}\\ 
(\F_i\mathbf{P}_i)^*\otimes \I_{M_j} & \I_{M_i}\otimes(\F_j\mathbf{P}_j)^*
\end{array}\right).
\end{equation}
Denote $\A_i:=\mathbf{P}_i^T\F_i\in\mathbb{C}^{L_i\times M_i}$ and $\A_j:=\mathbf{P}_j^T\F_j\in\mathbb{C}^{L_j\times M_j}$. Then, by assumption, it holds that either $\text{rank}(\A_i)=M_i$ and $\text{rank}(\A_j)=M_j$ or $\text{rank}(\A_i)=L_i$ and $\text{rank}(\A_j)=L_j$.
Let $\x=[\x_1^T\,\x_2^T]^T$ be such that $\M\x=0$ and show that $\x=0$.
Since $\M\x=0$, it holds
\begin{equation*}
\left\{\begin{array}{l}
(\I_{L_i}\otimes \A_j)\x_1 - (\A_i\otimes\I_{L_j})\x_2 = 0\\
(\A_i^H\otimes \I_{M_j})\x_1 + (\I_{M_i}\otimes\A_j)\x_2 = 0.
\end{array}\right.
\end{equation*}
Let $\X_1$ and $\X_2$ be matrices such that $\text{vec}(\X_1) = \x_1$ and $\text{vec}(\X_2) = \x_2$. Then
\begin{equation*}
\left\{\begin{array}{l}
\A_j\X_1 - \X_2\A_i^T = 0\\
\X_1\A_i^*+\A_j^H\X_2 = 0.
\end{array}\right.
\end{equation*}
Multiplying the first equation by $\A_j^H$ and the second by $\A_i^T$, and summing them up, we get $\A_j^H\A_j\X_1+\X_1(\A_i\A_i^H)^*=0$, which is a Sylvester's equation admitting a unique solution if $\A_j^H\A_j$ and $-(\A_i\A_i^H)^*$ have no common eigenvalues. On the other hand, the eigenvalues of $\A_j^H\A_j$ and $\A_i\A_i^H$ are real positive, so common eigenvalues of $\A_j^H\A_j$ and $-(\A_i\A_i^H)^*$ can only be $0$. However, this does not occur since by the assumptions either $\A_j^H\A_j$ or $\A_i\A_i^H$ is full rank. We can then conclude that $\X_1 = 0$, i.e., $\x_1=0$.
Similarly, $\x_2=0$, which ends the proof.
\end{proof}


\bibliographystyle{IEEEtran} 
\bibliography{refs}
\end{document}